\documentclass{article}
\usepackage[dvipsnames]{xcolor}
\usepackage{mathtools} 
\usepackage{amsthm}
\usepackage{thmtools}
\usepackage{enumerate}
\usepackage{amsmath}
\usepackage{amssymb}
\usepackage{amsfonts}
\usepackage{arxiv}
\newcounter{casecounter}
\newcommand{\Case}[1]{\refstepcounter{casecounter}\par\medskip\noindent\textbf{Case \thecasecounter} (#1):\ \ignorespaces
}
\usepackage{xspace}
\usepackage{color}
\usepackage{url}
\usepackage{comment}
\usepackage[colorlinks]{hyperref}
\hypersetup{
 colorlinks={true},
  linkcolor=[rgb]{0.3,0.3,0.6},
  citecolor=[rgb]{0.2, 0.6, 0.2},
  urlcolor=[rgb]{0.6, 0.2, 0.2}
}
\usepackage[capitalise, nameinlink]{cleveref}
\usepackage{bm}
\usepackage{bbm}
\usepackage{times}
\usepackage{enumitem}
\newcommand{\ineqref}[1]{\hyperref[#1]{Inequality~\ref*{#1}}}
\newtheorem{theorem}{Theorem}[section]
\newtheorem{thm}{Theorem}[section]

\newtheorem{lemma}[thm]{Lemma}
\newtheorem{corollary}[thm]{Corollary}
\newtheorem{claim}[thm]{Claim}

\newtheorem{remark}[thm]{Remark}

\newtheoremstyle{noparen}
{3pt}{3pt}
{}{}
{\bfseries}{:}
{12pt}{\thmname{#1}\thmnumber{ #2}\ \textnormal{\thmnote{ #3}}}
\newtheorem{problem}[thm]{Problem}
\renewcommand\b{\bm{\beta}}
\renewcommand\a{\bm{\alpha}}
\newcommand\F{\mathbb{F}}
\newcommand\R{\mathbb{R}}
\newcommand\C{\mathbb{C}}
\newcommand\Z{\mathbb{Z}}
\newcommand\Q{\mathbb{Q}}

\def\eps{\varepsilon}

\def\rhs{\text{RHS}}

\def\poly{\mathrm{poly}}

\def\NP{\mathsf{NP}}

\def\HN{\mathsf{HN}}
\def\etr{\exists\R}
\def\utr{\forall\R}
\def\lin{\mathsf{LIN}}
\def\bineg{\mathsf{Biquadratic Negativity}}
\def\proj{\mathsf{PolyProj}}
\def\ss{\mathsf{SparseShift}}

\def\coNP{\mathsf{coNP}}
\def\VP{\mathsf{VP}}
\def\VNP{\mathsf{VNP}}
\def\Perm{\mathsf{Perm}}
\def\Det{\mathsf{Det}}
\def\PH{\mathsf{PH}}

\DeclarePairedDelimiter{\norm}{\lVert}{\rVert}
\newcommand{\PSPACE}{\mathsf{PSPACE}}

\global\long\def\paren#1{\left(#1\right)}\newcommand{\AM}{\mathsf{AM}}
\newcommand{\defeq}{\coloneqq}

\title{Problems from Optimization and Computational Algebra Equivalent to Hilbert's Nullstellensatz}

\usepackage{authblk}

\author[1]{Markus Bläser\thanks{Email: \texttt{mblaeser@cs.uni-saarland.de}.}}

\author[2]{Sagnik Dutta\thanks{Email: \texttt{sadutta@mpi-inf.mpg.de}.}}

\author[3]{Gorav Jindal\thanks{Email: \texttt{gjindal@mimuw.edu.pl}.}}

\affil[1]{Saarland University, Saarland Informatics Campus, Saarbr\"ucken, Germany}
\affil[2]{Max Planck Institute for Informatics, Saarland Informatics Campus, Saarbr\"ucken, Germany}
\affil[3]{University of Warsaw, Warsaw, Poland}
\date{}

\date{\vspace{-6ex}}
\begin{document}

\maketitle

\begin{abstract}
Solving polynomial systems is a powerful tool in optimization and computational algebra. The most efficient algorithms for many problems from optimization or algebra often stem from formulating these problems as a system of polynomial equations. This is due to the fact that there are astonishing algorithms for deciding feasibility of polynomial systems, like Koiran's $\AM$-algorithm over the complex numbers (assuming GRH) or Renegar's $\PSPACE$-algorithm over the reals.

To formalize this, Blum, Shub, and Smale \cite{blum1998complexity} introduced Hilbert's Nullstellensatz Problem: Hilbert's Nullstellensatz Problem $\HN_R$ over some ring $R$ asks whether a given set of multivariate polynomials over $R$ has a common solution in $R$. $\HN_{\F_2}$ defines the class $\NP$, since every Boolean formula can be arithmetized appropriately. Over the real numbers,
$\HN_\R$ is equivalent to the existential theory of the reals $\etr$, which is astonishing, since in $\etr$, we can define semi-algebraic sets. We can view $\HN_R$ as a parameterized complexity class: for every ring $R$, we get a class by taking the downward closure of $\HN_R$ under polynomial-time many-one reductions.
In this work, we show that for many important problems from optimization and algebra, formulating them as a system of polynomials equations is optimal, since we can reduce Hilbert's Nullstellensatz to them.

We first show that two well-known problems from algebra are equivalent to or at least as hard as Hilbert's Nullstellensatz over fields. The first one is the Affine Polynomial Projection Problem, which given two polynomials, asks whether one of them can be transformed into the other by an affine linear projection of the variables. Kayal \cite{kayal} proves that this problem is $\NP$-hard. Here we improve this lower bound by showing that it is harder than $\HN_F$ for any field $F$. As a corollary, we obtain that this problem is complete for the existential theory of the reals, which answers an open question asked in  \cite[A14, open question]{SCM}. The second problem is the Sparse Shift Problem, which asks whether for a given polynomial, there is an affine shift that reduces the number of monomials. For integral domains $R$ that are not fields, Chillara, Grichener, and Shpilka \cite{chillara} showed that this problem is hard for $\HN_R$. They left it as an open problem whether this is also true for fields. We here resolve this question in the positive. Over infinite fields $F$, where $\HN_F$ is complete for $\NP_F$ (in the BSS model), we show that the Sparse Shift Problem is even equivalent to $\HN_F$.

Second, we turn to an important case of Hilbert's Nullstellensatz over the real numbers. Real-stable polynomials have been a successful tool in mathematics and computer science in the recent years. For instance, the solution to the Kadison-Singer problem by Marcus, Spielman, and Srivastava \cite{zbMATH06456012} is deeply connected to real stable polynomials, and so is the improvement of the approximation performance of metric TSP by Karlin, Klein, and Oveis Gharan \cite{KNO}, just to give two examples. We here prove that testing whether a given polynomial is real stable is equivalent to the complement of $\HN_\R$, or equivalently, to the universal theory of the reals $\utr$. We show that the same is true for testing convexity and testing hyperbolicity, as well as for testing whether a biquadratic form is nonnegative, completely settling the complexity of all of these problems.
\end{abstract}

\section{Introduction}
\label{sec:intro}

Solving systems of polynomial equations is a central technique in both optimization and computational algebra, serving as a unifying framework for a wide range of algorithmic problems. Many of the most efficient algorithms in these domains are built upon the insight that complex problems---ranging from feasibility in nonlinear programming to questions in algebraic geometry---can be reduced to the task of solving a polynomial system.

This approach is powerful because it takes advantage of the existence of surprisingly efficient algorithms for deciding the solvability of polynomial systems. Over the complex numbers, Koiran \cite{koiran} demonstrated that the feasibility of polynomial systems lies in the class $\AM$ (Arthur-Merlin games), assuming the Generalized Riemann Hypothesis (GRH). This result (conditionally) places the problem in a probabilistic class that is believed to be much smaller than $\PSPACE$. Canny \cite{canny} and Renegar \cite{DBLP:journals/jsc/Renegar92}  developed $\PSPACE$-algorithms for deciding the feasibility over real numbers. These results highlight that, despite the apparent algebraic and geometric complexity of polynomial systems, their computational complexity can be surprisingly low. As a result, formulating problems as polynomial systems not only provides a unifying language, but also yields algorithms  whose complexity is as low as possible. Note that in this context, even an  upper bound like $\PSPACE$ is surprisingly low, since real numbers are infinite objects and until the work of Tarski \cite{traski}, it was not even clear whether this problem was decidable at all. 
To formalize this, Blum, Shub, and Smale, see \cite{blum1998complexity}, introduced Hilbert's Nullstellensatz Problem: 
For an arbitrary ring $R$, the problem is defined as follows:

\smallskip
$\HN_{R}$ (Hilbert's Nullstellensatz over ring $R$) : Given a set of polynomial equations $f_1 = 0, \cdots, f_t=0$ where each $f_i$ is from $R[x_1,\cdots, x_n]$, decide if there exists a vector $(a_1,\cdots,a_n) \in R^n$ such that $f_i(a_1,\cdots,a_n) = 0$ for all $i \in [t]$.
\smallskip

This problem is also called the Polynomial Feasibility Problem, but in the present paper, we will use the term Hilbert's Nullstellensatz as coined by Blum, Shub and Smale.
When $R$ is an arbitrary ring, the coefficients of the input polynomials might not be efficiently representable on a Turing machine. Therefore, it makes more sense to use the Blum-Shub-Smale (BSS) model of computation where the registers can store arbitrary elements of the ring and each ring operation is assumed to cost unit time. However, often it's possible to consider the problem on the Turing machine model as well. If $R$ is a finite field, then the coefficients of the input polynomials will have constant bit complexity. If $R = \R$ or $\C$, we can assume the coefficients of the inputs to be integers. For other rings, we can restrict the coefficients to be ring elements with small bit complexity. In such settings, $\HN_{\F}$ is known to be $\NP$-hard for all fields $\F$.

The problem of $\HN$ over the real field $\R$ is particularly interesting because it is complete for the well-studied complexity class $\etr$. This class captures the complexity of deciding the truth of sentences in the existential theory of the reals (ETR). Formally, the ETR language consists of all true sentences of the form $\exists x_1, \dots, x_n \ \phi(x_1, \dots, x_n)$, interpreted over the reals, where $\phi$ is a quantifier-free formula written as a Boolean combination of atomic formulas involving the constants $0,1$, variables, the operations $+, \cdot$, and the comparison operators $<, >, \leq, \geq, =$. The class $\etr$ is the set of all languages that are polynomial-time many-one reducible to the ETR language.  Similarly, there is also the universal theory of reals, defined by its complete problem—the negation of $\HN_\R$, i.e., the problem of deciding whether a system of polynomials has no solutions over the reals. This theory corresponds to the class $\utr$, which captures the complexity of deciding the truth of universally quantified sentences over the reals. Over the past decade, these two classes have received significant attention in algorithms and complexity theory due to their ability to precisely characterize the complexity of numerous natural problems in fields such as algebra, computational geometry, game theory, and machine learning, among others. For a comprehensive account of $\etr$-complete and $\utr$-complete problems, see the survey \cite{SCM}.

For $\HN_{\R}$, the best known upper bound is due to Canny \cite{canny} and Renegar \cite{DBLP:journals/jsc/Renegar92}, who placed it in $\PSPACE$. Similarly, the best known unconditional upper bound for $\HN_{\C}$ is also $\PSPACE$. However, under the Generalized Riemann Hypothesis, Koiran \cite{koiran} showed that $\HN_{\C}$ lies within the second level of the polynomial hierarchy ($\PH$), and even $\AM$. Matiyasevich \cite{mati} established that $\HN_{\Z}$ is undecidable, while the decidability of $\HN_{\Q}$ remains a longstanding open problem. These results illustrate that Hilbert's Nullstellensatz is generally much harder than $\NP$. As a result, strengthening $\NP$-hardness (or $\coNP$-hardness) results by proving equivalence to Hilbert’s Nullstellensatz (or its negation) significantly enhances our understanding of the complexity of the concerned problems and also shows that formulating them as systems of polynomial equations is optimal.

In this paper, we go beyond previously known $\NP$-hardness and $\coNP$-hardness results. We prove that many well-known open problems—previously only known to be $\NP$-hard or $\coNP$-hard—are in fact $\etr$-complete or $\utr$-complete. This establishes a significantly higher level of complexity for these problems. See \cref{sec:our-results} for the precise statements of our results.

\subsection{Affine Polynomial Projection Problem} One of the key questions in Algebraic Complexity Theory is the problem of checking the equivalence of polynomials under affine transformations. For a field $\F$, the problem is defined as follows:

\begin{problem}[$\proj_{\F}$ (Affine polynomial projection)]
 Given polynomials $f \in \F[y_1,\cdots,y_m],
 g\in  \F[x_1,\cdots,x_n]$, decide if there exists an $m \times n$ matrix $A$  and a vector $b \in \F^m$ such that $f(A\mathbf{x}+b) = g(\mathbf{x})$.
\end{problem}

The central problem of $\VP$ vs $\VNP$ in Algebraic Complexity Theory is an instance of $\proj$. We will have $\VP \neq \VNP$ if the permanent polynomial $\Perm_m = \sum\limits_{\pi \in S_m}\prod\limits_{i=1}^m x_{i\pi(i)}$ cannot be written as an affine projection of the determinant polynomial
$\Det_n = \sum\limits_{\pi \in S_n}\text{sign}(\pi)\prod\limits_{i=1}^n x_{i\pi(i)}$ whenever $n = 2^{(\log m)^{O(1)}}$ \cite[Chapter 2]{burgisser}. The classical problem of matrix multiplication is also an instance of $\proj$. Matrix multiplication will have an $\tilde{O}(n^2)$ algorithm if the matrix multiplication polynomial $\mathsf{Mat}_n = \sum_{1\leq i,j,k \leq n} x_{ij}y_{jk}z_{ki}$ can be written as an affine projection of the sum of products polynomial $\mathsf{SP}_m = \sum_{i=1}^m x_{i1}x_{i2}x_{i3}$ for $m = \tilde{O}(n^2)$ \cite{STRASSEN1969, BI11}.

Kayal \cite{kayal} showed that $\proj_{\F}$ is $\NP$-hard in general. In this paper, we improve his result by showing that it is $\HN_{\F}$-hard for all fields $\F$. When $\F = \R$, it implies that the affine projection problem is $\etr$-complete, which solves \cite[A14, open question]{SCM}.

\subsection{Finding Sparse Shift of Polynomials} Since the problem of polynomial equivalence under affine projections is hard in general, different restricted formulations of the problem have been considered. Kayal in \cite{kayal} showed randomized polynomial-time algorithms for $\proj$ when $f$ is the permanent/determinant polynomial and $A$ is required to be full rank. If $A$ is required to be the identity matrix, then we obtain the following problem:
\begin{problem}[Shift Equivalence Testing (SET)]
Given $f,g \in R[x_1,\cdots,x_n]$, decide if there exists a shift $b \in R^n$ such that $f(\mathbf{x}+b) = g(\mathbf{x})$.   
\end{problem}
 For polynomials of degree $d$ and $n$ variables, Grigoriev \cite{GRIGORIEV1997} gave three algorithms for SET: a deterministic algorithm for characteristic $0$, a randomized algorithm for large enough characteristic $p>0$ and a quantum algorithm for characteristic $2$. All these algorithms run in time polynomial in $\binom{n+d}{d}$. Dvir, Oliviera and Shpilka \cite{DOS} showed that given blackbox access to $f$ and $g$, there is a randomized algorithm of running time $\poly(n,d,s)$ where $n$ is the number of variables, $d$ is the degree bound and $s$ is the size bound on the circuits for $f$ and $g$. The only randomized part of their algorithm is where they solve instances of Polynomial Identity Testing (PIT) and their algorithm can be derandomized if and only if PIT can be derandomized.

A variant of SET asks if a given polynomial $f \in \F[x_1,\cdots,x_n]$ has a shift $b \in \F^n$ such that $f(\mathbf{x}+b)$ has at most $t$ monomials for some parameter $t$. Such a shift is called a $t$-sparse shift. Lakshman and Saunders \cite{LS96} considered this problem for univariate polynomials over fields of characteristic zero and produced $\poly(t,d)$ time algorithms. Grigoriev and Lakshman extended this result to multivariate polynomials. Their algorithm computes $t$-sparse shifts for $n$-variate polynomials in deterministic running time $(nt)^{O(n^2)}$. After two decades of no improvement over this exponential time algorithm, Chillara, Grichener and Shpilka \cite{chillara} were motivated to study the hardness of this problem and they considered the following slightly more general version of it:
 \begin{problem}[$\ss_{R}$ (Sparsification of polynomial via shift)]
Given a polynomial $f \in R[x_1,\cdots,x_n]$, decide if there exists a vector $(a_1,\cdots,a_n) \in R^n$ such that $f(x_1+a_1,\cdots,x_n+a_n)$ has strictly fewer monomials  than $f(x_1,\cdots,x_n)$.     
 \end{problem}

They showed in \cite{chillara} that $\ss_R$ is $\HN_R$-hard when $R$ is an integral domain that is not a field. Their technique crucially relies on the non-invertibility of one element inside the integral domain. In this paper, we circumvent that problem and show that $\ss_\F$ is $\HN_\F$-hard for all infinite fields $\F$. 

\subsection{Convexity of Polynomials and Biquadratic Non-negativity}\label{subsec:convexbiquad}

Convexity plays a key role in modern mathematical optimization, often serving as a criterion for an optimization problem's tractability. Many combinatorial optimization problems can be formulated as polynomial optimization problems, where the objective function is a polynomial and the constraints are also described by polynomials. However, determining whether a given polynomial optimization problem exhibits convexity is not always straightforward. Convexity verification has important practical implications. In polynomial global minimization, verifying convexity ensures that every local minimum is global, simplifying optimization. We recall the formal definition of convexity now. A polynomial function $f: \mathbb{R}^n \to \mathbb{R}$ is said to be \textit{convex} if it satisfies the convexity condition:
\[
f(\lambda x + (1-\lambda) y) \leq \lambda f(x) + (1-\lambda) f(y), \quad \forall x, y \in \mathbb{R}^n, \forall \lambda \in [0,1]\ .
\]
Equivalently, $f(x)$ is convex if and only if its Hessian matrix 
\[
H_f(x) = \nabla^2 f(x)
\]
is \textit{positive semidefinite} for all $x \in \mathbb{R}^n$.
The convexity of odd-degree polynomials can be easily determined: linear polynomials are trivially convex, while polynomials of odd degree greater than one can never be convex. Moreover, for quadratic polynomials, checking the convexity means checking the positive semidefiniteness of its constant Hessian matrix, which can be done in polynomial time. Thus, quartic (degree 4) polynomials present the first non-trivial case where determining complexity of convexity testing remained open for a long time. This question  was posed as one of seven open problems in complexity theory for numerical optimization in 1992 \cite{PS} and only in 2013, \cite{AOPT} made a major advance by showing this problem to be $\NP$-hard under Turing reductions and $\coNP$-hard under Karp reductions.
\begin{problem}[Convexity Testing of Quartic Polynomials]\label{prob:ctq}
Given an $n$-variate polynomial of degree 4, decide if it is convex.
\end{problem}
The hardness proof in \cite{AOPT} proceeds through a reduction from \cref{prob:ctq} to the problem of deciding non-negativity of biquadratic forms. A biquadratic form $b(x, y)$ is a polynomial in the variables $x = (x_1, \dots, x_n)$ and $y = (y_1, \dots, y_m)$ that can be written as:
\[
B(x, y) = \sum_{i \leq j, k \leq l} \alpha_{ijkl} x_i x_j y_k y_l.
\] for some $\alpha_{ijkl} \in \R$.
\begin{problem}[Nonnegativity Testing of Biquadratic Forms]\label{prob:biqnn}
Given a biquadratic form \\ $B(x_1, \dots, x_n, y_1, \dots, y_m)$, determine whether it is non-negative, i.e., decide whether  
\[
B(x_1, \dots, x_n, y_1, \dots, y_m) \geq 0
\]
for all $(x_1, \dots, x_n, y_1, \dots, y_m) \in \mathbb{R}^{n+m}$.
\end{problem}

 \cite{LNQY} showed \cref{prob:biqnn} to be $\coNP$-hard under Karp reductions (and $\NP$-hard under Turing reductions) and this implied the hardness result for convexity testing in \cite{AOPT}. \cite{LNQY} actually treated a more general version of \cref{prob:biqnn} where one asks for the minimum objective value of a biquadratic form $b(x,y)$ over the unit bi-sphere $\{(x,y) \in \R^{2n}\ :\ \norm{x}=\norm{y}=1\}$. This biquadratic optimization problem is interesting in its own right. It arises from the strong ellipticity condition problem in solid mechanics \cite{Knowles1975, RP90, Simpson1983, Wang1996} and the fundamental quantum theoretic problem of determining if a quantum state is entangled \cite{Gurvits}.

For an upper bound, \cite{canny} showed that \cref{prob:ctq} is in $\PSPACE$. We make significant progress on both of these problems by showing their $\utr$-completeness, hence settling their complexity; see \cref{sec:our-results} for more details.

\subsection{Checking Hyperbolicity and Real Stability of Polynomials}

\textit{Real stable} polynomials have proved to be a successful tool in some of the major mathematical advances. For instance, their theory was used to resolve the well-known Kadison-Singer problem \cite{zbMATH06456012} and to construct infinite families of Ramanujan graphs \cite{doi:10.1137/16M106176X,10.1109/FOCS.2013.63}. These polynomials have also found application in  improvement of approximation algorithms for metric TSP \cite{KNO}.

A multivariate real polynomial $p$ is said to be \textit{real stable} if it has no roots in the open upper half of the complex plane. Equivalently, $p$ is \textit{real stable} if and only if the univariate polynomial  
\[
p(te_1 + x_1, te_2 + x_2, \dots, te_n + x_n)
\]
is real-rooted for all choices of $e_1, \dots, e_n > 0$ and $x_1, \dots, x_n \in \mathbb{R}$.  It can also be shown as an exercise that a real univariate polynomial is \textit{real stable} if and only if it is real-rooted, meaning all its roots are real. This concept of \textit{real stability} plays a fundamental role in probability theory, control theory, and various other areas of mathematics.  

Hyperbolic polynomials generalize the concept of \textit{real stable} polynomials. A homogeneous $n$-variate polynomial $p$ is said to be \textit{hyperbolic with respect to a vector} $e\in \R^n$ if $p(e) > 0$ and $\forall x \in \R^n$, the univariate polynomial $p(x+te) \in \R[t]$ is real-rooted. Hyperbolic polynomials are of great interest in optimization theory due to their generality. They define barrier functions for interior point methods in hyperbolic optimization, which is a powerful generalization of semidefinite programming. In this paper, we look at real stable and hyperbolic polynomials from a computational complexity perspective. Specifically, we aim to characterize the complexity of the following problems:  

\begin{problem}[Real Stability Testing]\label{prob:rstest}
Given a real multivariate polynomial $p$, decide whether $p$ is \textit{real stable}.
\end{problem}  

\begin{problem}[Hyperbolicity Testing]\label{prob:hypertest}
Given a homogeneous multivariate polynomial $p$ and a direction $e \in \mathbb{R}^n$, decide whether $p$ is hyperbolic with respect to $e$.
\end{problem}  

A polynomial-time algorithm for the univariate variant of \cref{prob:rstest} follows immediately from \cite{sturm1835memoire}. The bivariate case of \cref{prob:rstest} was shown to be decidable in polynomial time by Raghavendra, Ryder and Srivastava \cite{raghavendra17itcs}.  Then it is  a natural question to ask about the complexity of \cref{prob:rstest} and \cref{prob:hypertest} for multivariate polynomials. To this end, it was shown by Saunderson \cite{saunderson}
 that \cref{prob:hypertest} for cubic multivariate polynomials is $\coNP$-hard.

For \cref{prob:rstest}, Chin \cite{chin2024}  established the $\coNP$-hardness of determining the real stability of homogeneous polynomials. Specifically, it is proven in \cite{chin2024}   that deciding real stability for polynomials of degree $d\geq 3$ with rational coefficients is $\coNP$-hard.  For a more detailed discussion, see \cite{raghavendra17itcs, chin2024} and the references therein.  In this paper, we significantly improve these $\coNP$-hardness results and prove that both of these problems are $\utr$-complete for quartic polynomials. Note that the prior works established $\coNP$-hardness results for hyperbolicity and real stability of cubic polynomials whereas our reductions achieve $\utr$-hardness for quartic polynomials. We leave it open whether $\utr$-hardness of these problems can be shown for cubic polynomials as well.

\section{Our Results}
\label{sec:our-results}

Before stating our results, we should discuss the issue of how the input polynomials are provided. We are mainly concerned with three kinds of white-box representations of the input polynomials: the \textit{dense representation} where we are given a list of all the coefficients, the \textit{sparse representation} where we are given a list of all the non-zero coefficients along with the corresponding monomials, and the \textit{arithmetic circuit representation} where we are given a circuit computing the polynomial. For more details, see \Cref{sec:prelim}.

Our first result gives a reduction from Hilbert's Nullstellensatz to the Sparse Shift Problem under all of the above representations.
\begin{restatable}{theorem}{restatablethmsparseshift}\label{sparseshift-thm}
  For all infinite fields $\F$, $\ss_{\F}$ is $\HN_{\F}$-hard in any of the white-box representations.
\end{restatable}
This extends the result of \cite{chillara} where they showed $\ss_{R}$ is $\HN_{R}$-hard for all integral domains $R$ that are not fields. When $R$ is a non-field, it contains an element without any inverse and this element was crucially used in the reduction of \cite{chillara} to build the polynomial for the $\ss$ instance. We produce a reduction that doesn't rely on the existence of such non-invertible elements.

\begin{remark}
    Our proof of \Cref{sparseshift-thm} actually shows that for a field $\F$ of size at least $4n^4$, $\HN_{\F}$ on $n$ variables can be reduced to $\ss_{\F}$ in $\poly(n,s)$ time where $s$ is the total size of the white-box representations of the input polynomials. Therefore, given an $\HN$ instance on $n$ variables over a finite field $\F_q$, we can consider an extension field $K_{(n)}$ of $\F_q$ of size at least $4n^4$. If we can show that $\HN$ over $\F_q$ reduces to $\HN$ over $K_{(n)}$, then we obtain the interesting poly-time reduction of $\HN_{\F_q}$ on $n$ variables to $\ss_{K_{(n)}}$. This is indeed possible since a system $\{f_i(x_1,\cdots,x_n) = 0\}_{i=1}^t$ has a solution over $\F_q$ if and only if the system $\{f_i(x_1,\cdots,x_n) = 0\}_{i=1}^t \cup \{x_i^q - x_i = 0\}_{i=1}^n$ has a solution over $K_{(n)}$.
\end{remark}

Under the sparse or dense representation, $\ss_R$ is in $\NP_R$, the analogue of $\NP$ for the BSS model of computation. This is because given a vector $(a_1,\cdots,a_n)$ and the sparse representation of a polynomial $g(x_1,\cdots,x_n)$, we can efficiently verify if $f(x_1+a_1,\cdots,x_n+a_n) = g(x_1,\cdots,x_n)$ using sparse polynomial interpolation and sparse polynomial identity testing. Thus, we have the following corollary of \Cref{sparseshift-thm}.

\begin{corollary}
   Let $\F$ be an infinite field such that $\HN_{\F}$ is $\NP_{\F}$-complete. Then, $\ss_{\F}$ is $\NP_{\F}$-complete and equivalent to $\HN_{\F}$.
\end{corollary}
 \cite{blum1989theory} showed that $\HN_{\F}$ is $\NP_{\F}$-complete when $\F = \R$ or $\C$. Thus, $\ss$ is complete for the existential theory of these fields. Our next result shows that the affine polynomial projection problem is harder than Hilbert's Nullstellensatz over all fields.

\begin{restatable}{theorem}{restatablethmpolyproj}\label{polyproj-thm}
  For all fields $\F$, the problem $\proj_{\F}$ is $\HN_{\F}$-hard under both sparse and arithmetic circuit representations.
\end{restatable}

Kayal proved in \cite{kayal} that this problem is $\NP$-hard. \Cref{polyproj-thm} improves upon this lower bound. We also have a reduction from $\proj$ to $\HN$ when the underlying field is $\R$ or $\C$. Given $f$ and $g$, we need to guess the matrix $A$ and the vector $b$ such that $f(A\mathbf{x}+b)-g(\mathbf{x})$ is the zero polynomial and then we use the observation that polynomial identity testing reduces to the existential theory over the reals or the complex numbers. Therefore, we obtain as a corollary of \Cref{polyproj-thm} that $\proj_{\F}$ is $\HN_{\F}$-complete when $\F=\R$ or $\C$. This solves \cite[A14, open question]{SCM} which asked if $\proj$ is complete for the existential theory of reals.

Ling, Nie, Qi and Ye \cite{LNQY} showed that it is $\coNP$-hard to compute the minimum value of a biquadratic objective function over the bi-sphere. With a minor adaptation, their proof also implies $\coNP$-hardness of deciding if a biquadratic form is non-negative. The following theorem shows that this problem is in fact complete for the universal theory of reals.

\begin{restatable}{theorem}{restatablethmbpsd}\label{bpsd-thm}
    Deciding non-negativity of a biquadratic form is $\utr$-complete.
\end{restatable}

In \cite[Lemma 2.8]{DBLP:journals/mst/SchaeferS24}, Schaefer and Stefankovic showed through an involved construction that the problem of deciding non-negativity of a degree $6$ polynomial is $\utr$-hard. While they also described how to reduce the degree in the construction to 4 \cite[Remark 2.9]{DBLP:journals/mst/SchaeferS24}, the polynomial constructed is not a biquadratic form. Since biquadratic forms have a very restricted structure, we need an even more intricate construction for \Cref{bpsd-thm}.

\cite{AOPT} reduced the problem of deciding non-negativity of biquadratic forms to deciding convexity of quartic polynomials, thereby obtaining $\coNP$-hardness of the latter problem from that of the former. Using their reduction, we get $\utr$-hardness of convexity testing as a corollary of \Cref{bpsd-thm}.

\begin{theorem}\label{convexity-thm}
    Deciding convexity of a degree $4$ polynomial is $\utr$-complete.
\end{theorem}

Note that the papers \cite{LNQY} and \cite{AOPT} claim to prove only $\NP$-hardness of the problems of deciding non-negativity of biquadratic forms and deciding convexity of quartic polynomials, respectively. However, their proofs use Turing reductions instead of polynomial-time many-one reductions (Karp reductions), which we are more interested in. In fact, their arguments imply $\coNP$-hardness of these problems under Karp reductions.

Saunderson \cite{saunderson} proved that the problem of deciding if a cubic homogeneous polynomial is hyperbolic with respect to a vector is $\coNP$-hard. We show that if you allow quartic polynomials in the reduction, then this problem is $\utr$-hard as well. 

\begin{theorem}\label{hyperbolic-thm}
    Given a homogeneous quartic polynomial $p$ and a vector $e$, the problem of deciding if $p$ is hyperbolic with respect to $e$ is $\utr$-complete.
\end{theorem}

Chin \cite{chin2024} showed that testing hyperbolicity of the polynomial in Saunderson's construction reduces to testing real stability of a polynomial, thereby leveraging Saunderson's result to demonstrate $\coNP$-hardness of testing real stability. Using Chin's approach, we get $\utr$-hardness of real stability testing as a corollary of \Cref{hyperbolic-thm}.

\begin{theorem}\label{stability-thm}
    Given a homogeneous quartic polynomial $p$, the problem of deciding if $p$ is real stable is $\utr$-complete.
\end{theorem}

\begin{remark}
 For \Cref{bpsd-thm,convexity-thm,hyperbolic-thm,stability-thm}, we deal with constant degree polynomials, hence all the white-box representations have the same complexity up to polynomial factors and it does not matter which one we use.
\end{remark}

 \subsection{Proof Idea for \Cref{sparseshift-thm}}
 Throughout this paper, the terms ``monomial count'' and ``number of monomials'' are used interchangeably to denote the number of nonzero terms in a polynomial. 

 Our goal is to establish a reduction from the $\HN$ problem to the $\ss$ problem by constructing a polynomial $Q_S$ such that $Q_S$ has a sparse shift if and only if the given system $S$ of polynomial equations has a solution. 

 In \cite{chillara}, they show such a reduction when the underlying ring is an integral domain, but not a field, and hence contains an element $\gamma$ without any multiplicative inverse. For structured polynomial systems
\[
S: \{g_1(x_1,\cdots,x_N)=0,\cdots,g_t(x_1,\cdots,x_N)=0\},
\] 
they show that the polynomial
\[
  Q_S(x_0,\cdots,x_N,w_1,\cdots,w_t) := w_1\cdot g_1(\mathbf{x}) + \sum_{i=2}^t w_i \cdot (\gamma g_i(\mathbf{x}) + \sum_{k=0}^N x_k)
\]has a sparse shift if and only if $S$ has a solution. The non-invertibility of $\gamma$ helps in showing that after a shift, no cancellation of the monomials takes place if $S$ is not satisfiable. A straightforward idea for extending this approach to fields could be to treat $\gamma$ as a variable itself and try to make the same proof work. However, now $\gamma$ itself can be shifted as well, which presents a serious obstacle to the proof. Therefore, we needed new ideas and more complex machinery to design $Q_S$, which we briefly sketch below.

First, we utilize the \textit{Structure Lemmas} from \cite{chillara} to transform $S$ into an equivalent system where all polynomials have degree at most 2, and at most one equation has a nonzero constant term. This structured form simplifies the analysis and polynomial construction. Next, we introduce a new set of variables $\a$ corresponding to the monomials of degree at most 2. The original system is then rewritten in terms of these new variables, leading to an equivalent system $S_1 \cup S_2 \cup S_3$, where:
\begin{itemize}[nosep]
    \item $S_1$ captures the original equations reformulated in terms of the new variables.
    \item $S_2$ encodes quadratic consistency constraints.
    \item $S_3$ ensures a specific sum condition on the variables.
\end{itemize}

It is then shown that the original polynomial system $S$ has a solution if and only if the system $S_1 \cup S_2 \cup S_3$ has one. We then introduce two additional sets of variables, $\b$ and $Y$.

Next, we construct a polynomial $P_S$ over the expanded set of variables $(\a, \b, Y)$ such that a sparsifying shift (satisfying certain conditions) exists for $P_S$ if and only if $S$ has a solution. This is achieved by analyzing how the monomial count of $P_S$ changes under shifts in its variables. By carefully tracking these changes, we establish conditions under which a shift reduces monomial count.

Finally, we construct $Q_S$ such that the existence of a sparsifying shift for $P_S$ (under a particular linear constraint) directly corresponds to the existence of a sparsifying shift for $Q_S$ in the unrestricted setting, thereby completing the reduction.

One of the key ideas in our reduction is \textit{amplification of the monomial count increase} in different parts of the polynomial under shift. This allows us to enforce the satisfiability of an equation in the following way: if the equation is not satisfied by a shift, then the monomial count of the corresponding part increases and we can amplify this monomial count increase heavily so that there is an incentive against not satisfying the equation. This amplification is achieved as follows:  let $P$ be a polynomial in the variables $x_1, \dots, x_m$. If we consider the polynomial
\[
Q := \paren{\sum_{i\in[k]}y_i} P(x_1, \dots, x_m),
\]
then, any shift that decreases the monomial count of $Q$ can clearly be assumed to leave the $y$ variables unchanged. Moreover, if a shift in the $x$ variables increases the monomial count of $P$ by $\ell$, then the same shift increases the monomial count of $Q$ by $\ell k$.

We apply this idea to various parts of the polynomial constructed in our reduction to penalize the number of monomials that can be added by a shift. The full proof, which is significantly more intricate and involves substantial technical challenges, is provided in \Cref{sec:res}. It requires carefully constructing the polynomial $P_S$ and thoroughly analyzing how shifts affect its monomial count. For a complete understanding, we refer the reader to the detailed exposition in \Cref{sec:res}.

\subsection{Proof Idea for \Cref{polyproj-thm}}

We are given a system $S = \{f_1(X) = 0, f_2(X)=0, \cdots, f_t(X)=0\}$ of polynomial equations defined over the variable set $X = \{x_1,\cdots,x_n\}$. Using \Cref{structure-lemma}, we can assume that each $f_i$ is of degree at most $2$. We introduce new variables $x_0, w_1, \cdots, w_t, z_1, \cdots, z_t$. Our polynomials $f,g$ for the $\proj$ instance will be defined over the variable vector\ $Y = (x_0, \cdots, x_n, w_1,\cdots,w_t, z_1,\cdots,z_t)$.  For two increasing sequences of integers $d_1 < \cdots < d_t$ and $D_1 < \cdots < D_n$ satisfying $2d_t < D_1$, we define $f,g$ as below:
\begin{align*}
    &f(Y) = \sum_{i=1}^t w_i^{d_i}(f_i(X) + z_i^{d_i}) + x_0 + \sum_{i=1}^n x_i^{D_i} \hspace{1cm}\text{ and } \\
    &g(Y) = \sum_{i=1}^t w_i^{d_i} z_i^{d_i} + x_0.
\end{align*}
We  show that $S$ has a solution if and only if there exist a matrix $A \in \mathbb{F}^{|Y|\times|Y|}$ and a vector $b \in \mathbb{F}^{|Y|}$ such that $f(AY+b) = g(Y)$. Note that $AY+b$ is simply a vector of affine linear polynomials over $Y$.

If $S$ has a solution $(a_1,\cdots,a_n)$, then the vector\ $Y' = (x_0-\sum\limits_{i=1}^n a_i^{D_i}, a_1,\cdots,a_n,w_1\cdots,w_t,z_1,\cdots,z_t)$ satisfies
\[
f(Y') = \sum\limits_{i=1}^t w_i^{d_i}(0 + z_i^{d_i}) + x_0 -\sum_{i=1}^n a_i^{D_i} + \sum_{i=1}^n a_i^{D_i} = g(Y).
\]
Thus, we can find a matrix $A$ and a vector $b$ such that $f(AY+b) = f(Y') = g(Y)$. 

Conversely, suppose that $Y' = (L_0, \cdots, L_n, P_1, \cdots, P_t, Q_1, \cdots, Q_t)$ is a vector of affine linear polynomials satisfying $f(Y') = g(Y)$. The main goal now is to show that $L_i$ must be constant for all $i$.

The proof proceeds by showing that each $L_i$ must be a constant to maintain the degree constraints imposed by $g(Y)$. The crucial observation is that the highest degree terms in $f(Y')$ must match those in $g(Y)$, and this forces each $L_i$ to be constant.
More specifically, for $i=n$, if $L_n$ were not constant, then $f(Y')$ would have terms of degree $D_n$, contradicting with the fact that $g(Y)$ has degree much smaller than $D_n$. Similarly, for smaller values of $i$, we use the fact that the highest degree term must be controlled at each step to show that $L_i$ is constant.

Since $L_i$'s are all constants, we now analyze $P_i Q_i$. We establish that $P_i Q_i = w_i z_i$ using the fact that any additional terms in the product would lead to monomials that are not present in $g(Y)$. By induction, this property propagates across all indices. From $P_i Q_i = w_i z_i$, we deduce that $f_i(L_1, \dots, L_n) = 0$ for all $i$, meaning $(L_1, \dots, L_n)$ is a valid solution of $S$. 

\subsection{Proof Idea for \Cref{bpsd-thm}}

To prove \Cref{bpsd-thm}, we follow and generalize the ideas used in the proof of Lemma 2.8 in \cite{DBLP:journals/mst/SchaeferS24}. Specifically, we reduce an existing $\etr$-complete problem—determining whether a set of quadratic polynomials has a common root in the unit ball—to the problem of deciding whether a biquadratic form ever attains a negative value, we call this problem to be $\bineg$.

Recall that a biquadratic form over a variable set $X \sqcup Y$ is a degree $4$ polynomial such that each monomial of it contains exactly two $X$-variables and exactly two $Y$-variables. Let us call a polynomial \textit{semi-biquadratic} if every monomial of it  contains at most two $X$-variables and at most two $Y$-variables. Our construction of the final biquadratic form from the system of quadratic polynomials goes through several stages. At each stage, we maintain a semi-biquadratic polynomial whose negativity is to be checked. Only at the final step, we appropriately homogenize the semi-biquadratic polynomial to make it a biquadratic form.

First, we establish a generalization of Lemma 2.7 in \cite{DBLP:journals/mst/SchaeferS24}, stated as \Cref{chain-lemma}. This lemma provides a \textit{rapid decay bound} for the sequence $y_m$ in terms of its initial value $y_0$. It states that if the sequence $\{y_k, z_k\}$ satisfies a given inequality, then $y_m$ decreases \textit{doubly exponentially} as a function of $m$.  

Now, suppose $f_1, f_2, \dots, f_t$ are quadratic polynomials over a variable set $x$, and we want to check whether they have a common root in the unit ball. We first construct a semi-biquadratic polynomial $g(x,w)$ whose solutions correspond to common solutions of these quadratic polynomials, where $w$ is a new set of variables. This polynomial is constructed such that if the polynomials $f_i$ have a common solution inside the unit ball, then $g$ has a root $(x,w) \in \mathbb{R}^{2n+2}$ satisfying $\sum_{i=1}^n x_i^2 \leq 1$. Conversely, if the polynomials $f_i$ do not have a common solution inside the unit ball, we show that the minimum value of $g(x,w)$ on the semi-algebraic set $S := \{(x,w) \in \R^{2n+2}:\ 
  \sum_{i=1}^n (x_i^2+w_i^2) \leq 3\}$ can be lower-bounded by a doubly exponentially small number.  

Following the proof strategy of Lemma 2.8 in \cite{DBLP:journals/mst/SchaeferS24}, we now construct a new semi-biquadratic polynomial $h(x,w,y,z)$, which consists of $g(x,w)$ along with a modified version of the inequality used in \Cref{chain-lemma}.  

Next, we prove that $h(x,w,y,z)$ attains negative values if and only if the polynomials $f_i$ have a common solution inside the unit ball. The forward direction—showing that a common solution of the $f_i$'s implies that $h(x,w,y,z)$ takes a negative value for some $(x,w,y,z)$—is straightforward. Conversely, if $h$ takes negative values, then using the \Cref{chain-lemma}, we establish that $y_m^2 < 2^{-2^m}$, which forces $g(x,w)$ to be \textit{extremely small} on the set $S$. Combining this with the previously established lower bound on the minimum value of $g(x,w)$ over $S$, we conclude that the polynomials $f_i$ have a common solution inside the unit ball.  

Finally, we homogenize the polynomial $h(x,w,y,z)$ using two variables $s,t$ to obtain a biquadratic form $Q(x,w,y,z,s,t)$, completing the reduction and proving the $\etr$-hardness of $\bineg$.

\subsection{Proof Idea for \Cref{hyperbolic-thm}}

The proof that hyperbolicity is in $\utr$ follows directly from the definition of hyperbolicity. Specifically, the condition that $p(x + te)$ is real-rooted for all $x$ can be decided in $\utr$. To see this, observe that one can efficiently construct the real and imaginary parts of the evaluation of a real polynomial at a complex point. This holds even when $p$ is given as an arithmetic circuit, ensuring that the problem remains in $\utr$.  

To establish the $\utr$-hardness of hyperbolicity, we reduce the problem of deciding the non-negativity of a biquadratic polynomial to the problem of deciding hyperbolicity. By \Cref{bpsd-thm}, the problem of deciding the non-negativity of a biquadratic polynomial is $\utr$-hard.  

Given a biquadratic form $Q$, we construct a homogeneous polynomial of degree 4 and a vector $e$ such that $p$ is hyperbolic with respect to $e$ if and only if $Q$ is non-negative. In particular, we define the polynomial $p(x_0, \dots, x_n)$ as:  
\[
p(x_0, \dots, x_n) = x_0^4 - \beta x_0^2 \norm{x}^2 + Q(x_1, \dots, x_n)
\]
for a suitably large constant $\beta$, and we take $e$ to be the vector $(1, 0, \dots, 0)$.  

By \Cref{subspace}, we know that $p$ is hyperbolic with respect to $e$ if and only if $B_{p,e}(x)$ is positive semidefinite for all $x$ with $x_0 = 0$, where $B_{p,e}(x)$ is the parametrized B\'ezoutian matrix defined in \Cref{subsec:hypandrspoly}.  Now, we compute $B_{p,e}(x)$ and use the Schur complement to derive the condition for hyperbolicity. A simple calculation shows that $p$ is hyperbolic with respect to $e$ if and only if  
\[
\forall x_1, \dots, x_n \in \mathbb{R}, \quad 0 \leq Q(x_1, \dots, x_n) \leq \frac{\beta^2}{4} \norm{x}^4.
\]
The proof is completed by choosing $\beta$ to be a sufficiently large constant so that the upper bound condition $Q(x_1, \dots, x_n) \leq \frac{\beta^2}{4} \norm{x}^4$ is always satisfied.

\section{Preliminaries}
\label{sec:prelim}

\subsection{Polynomials}

A \textit{polynomial} $f(x_1,\cdots,x_n)$ over a ring $R$ is a finite $R$-linear combination of monomials:
$$f(x_1,\cdots,x_n) = \sum_{(e_1,\cdots,e_n)\in I} \alpha_{e_1,\cdots,e_n}x_1^{e_1}\cdots x_n^{e_n}$$
where the index set $I$ is a finite subset of $\Z^n_{\geq 0}$ and the coefficients $\alpha_{e_1,\cdots,e_n}$ are elements of the ring $R$.

The degree of a monomial $x_1^{e_1}\cdots x_n^{e_n}$ is the sum $e_1+\cdots+e_n$. A polynomial is called \textit{homogeneous} if all of its monomials have the same degree.

A homogeneous degree-$4$ polynomial $Q(x_1,\cdots,x_n)$ is called a \textit{biquadratic form} if the index set $[n]$ can be partitioned into two subsets $S$ and $T$ such that $Q$ can be written as 
$$Q(x_1,\cdots,x_n) = \sum_{i,j \in S} \sum_{k,l \in T} \alpha_{ijkl} x_ix_jx_kx_l\ .$$

From the computational perspective, there are a number of models to represent polynomials. An algorithm may access the input polynomial as black-box or it may have access to one of the following white-box representations of the polynomial:
\begin{itemize}[nosep]
    \item In the \textit{dense representation}, an $n$-variate degree $d$ polynomial is given as a list of the coefficients of its $\binom{n+d}{d}$ monomials.
    \item The \textit{sparse representation} is more compact. Here, the polynomial is given as a list of pairs of monomials and coefficients. Hence, the list has length equal to the monomial count of the polynomial.
    \item In the \textit{arithmetic circuit representation}, we are given an arithmetic circuit computing the polynomial. The complexity of the algorithm is considered in terms of the bit length of the circuit's description.
\end{itemize}

\subsection{PSD-ness and Convexity}

We have the notion of positive-semidefiniteness for both polynomials as well as matrices. A polynomial $p \in \R[x_1,\cdots,x_n]$ is said to be \textit{positive semidefinite (PSD)} or \textit{non-negative} if $p(x) \geq 0$ for all $x \in \R^n$. On the other hand, a real symmetric matrix $A_{n \times n}$ is called PSD if $x^TAx \geq 0$ for all $x \in \R^n$. We denote this by $A \succeq 0$.

The \textit{Hessian} of a polynomial $f \in \R[x_1,\cdots,x_n]$ is the $n\times n$ polynomial matrix $H_f(x)$ defined by 
$[H_f(x)]_{ij} = \frac{\partial^2 f}{\partial x_i \partial x_j}$. We call a polynomial $f$ \textit{convex} if and only if $H_f(x) \succeq 0$ for all $x \in \R^n$. 

\subsection{Hyperbolic and Real Stable Polynomials}\label{subsec:hypandrspoly}

Let $\R[x_1,\cdots,x_n]_d$ denote the set of real homogeneous polynomials of degree $d$. A polynomial $p \in \R[x_1,\cdots,x_n]_d$  is said to be \textit{hyperbolic with respect to} $e \in \R^n$ if $p(e) > 0$ and $\forall x \in \R^n$, the univariate polynomial $p(x+te) \in \R[t]$ has only real roots.

If $p$ is hyperbolic with respect to $e$, then we can define the \textit{hyperbolicity cone}
$$\Lambda_+(p,e) = \{x \in \R^n \ |\ \text{all roots of $p(te-x)$ are positive}\}\ .$$
The following is a well-known result on hyperbolicity:
\begin{theorem}[\cite{garding}]\label{garding}
    If $p$ is hyperbolic with respect to $e$, then it is also hyperbolic with respect to all $x \in \Lambda_+(p,e)$. Moreover, $\Lambda_+(p,e)$ is the connected component of $\{x \in \R^n \ |\ p(x) \neq 0\}$ containing $e$.
\end{theorem}

For $p \in \R[x_1,\cdots,x_n]_d$ and $e \in \R^n$, the directional derivative of $p$ in the direction of $e$ is $D_e p(x) = \frac{d}{du} p(x+ue) \vert_{u=0}$. The \textit{parametrized Bezoutian} $B_{p,e}(x)$ is defined to be the $d\times d$ symmetric matrix with polynomial entries given by the following identity:
$$ \frac{p(x+te)D_ep(x+se) - p(x+se)D_ep(x+te)}{t-s} = \sum_{i,j=1}^d [B_{p,e}(x)]_{ij} t^{i-1}s^{j-1} $$
where $s,t$ are variables.

The following theorem relates the hyperbolicity of a polynomial to the PSD-ness  of the parametrized Bezoutian.
\begin{theorem}[\cite{saunderson}, Corollary 3.9]\label{subspace}
Let $p \in \R[x_1,\cdots,x_n]_d,\ e \in \R^n$ and let $W \subseteq \R^n$ be a subspace of codimension $1$ such that $e \not\in W$. Then, $p$ is hyperbolic with respect to $e$ if and only if $B_{p,e}(x) \succeq 0$ for all $x \in W$.
\end{theorem}

A polynomial $p \in \C[x_1,\cdots,x_n]$ is called \textit{stable} if $p(z_1,\cdots,z_n) \neq 0$ whenever each $z_i$ is a complex number with a positive imaginary part. Additionally, we say $p$ is \textit{real stable} if all the coefficients of $p$ are real. 

There is an equivalent definition of real stability which does not use the complex field. A polynomial $p \in \R[x_1,\cdots,x_n]$ is real stable if and only if for every point $a \in \R^n_{>0}$ and $b \in \R^n$, the univariate polynomial $q_{a,b}(t) = p(at+b)$ is not identically zero and has only real roots. If $p$ is homogeneous (of positive degree), then the above condition is equivalent to $p$ being hyperbolic with respect to $\R^n_{>0}$.

\section{Reduction from HN to SparseShift}
\label{sec:res}

In this section,  we prove \Cref{sparseshift-thm}. 
Given a system $S$ of $n$-variate polynomial equations over a field $\F$ of size more than $4n^4$, we will show how to construct a polynomial $Q_S$ such that $Q_S$ admits a sparse shift if and only if $S$ has a solution.

As mentioned in \Cref{sec:our-results}, \cite{chillara} proved \Cref{sparseshift-thm} over integral domains that are not fields. Crucially, their proof relies on the existence of a nonzero element $\gamma$ that is not invertible. A natural idea to extend their argument to fields would be to treat $\gamma$ as an additional variable. However, this introduces the complication that $\gamma$ itself can be shifted, making it difficult to preserve the desired hardness property. Consequently, the hardness proof in \cite{chillara} does not generalize to fields in a straightforward way.

In this section, we resolve this difficulty through a novel and technically involved construction. We will use the following lemma from \cite{chillara}, which lets us assume additional structure on the polynomial system $S$.

\begin{lemma}[\cite{chillara}, Lemmas $7$ and $8$]\label{structure-lemma}
    Let $S = \{f_i(X) = 0\}_{i=1}^r$ be a system of polynomial equations over an integral domain $R$ where each $f_i$ is given either as a list of coefficients or as an arithmetic circuit. In both cases, we can obtain in polynomial time a system $T$ of polynomial equations $\{g_j(Y) = 0\}_{j=1}^t$ such that 
    \begin{itemize}
        \item $S$ has a solution in $R^{|X|}$ if and only if $T$ has a solution in $R^{|Y|}$.
        \item Each $g_j(Y)$ is of degree at most $2$. 
        \item There is at most one $j \in [t]$ such that $g_j(Y)$ has a non-zero constant term.
    \end{itemize}
\end{lemma}

\begin{remark}
    Lemmas $7$ and $8$ of \cite{chillara} describe how to obtain a system $T$ satisfying only the first two conditions above. However, Section 3.2 of \cite{chillara} describes a simple trick to achieve the third condition as well: if $T$ contains the equations $g_1(Y) = 0, \cdots, g_t(Y)=0$ with constant terms $c_1 \neq 0,c_2,\cdots,c_t$ respectively, then obtain $T'$ from $T$ by simply updating each polynomial $g_j(Y)$ as follows: $g_j(Y) \leftarrow c_1g_j(Y)-c_jg_1(Y)$. Then, $T'$ has a solution if and only if $T$ does (since $R$ is an integral domain) and $T'$ contains only one polynomial with a non-zero constant term.
\end{remark}
Therefore, using \Cref{structure-lemma}, we can assume that we are given a system 
\[
S = \{f_1(X) = 0, f_2(X) = 0, \cdots, f_t(X) = 0\}
\]
of polynomial equations over the variable set $X = \{x_1, \cdots, x_n\}$ where each $f_i$ is of degree at most $2$ and $f_1$ is the only polynomial with a non-zero constant term. Let 
\[
M_X \defeq \{ x_j \mid 1 \le j \le n \} \cup \{ x_j x_k \mid 1 \le j \le k \le n \}.
\]
be the set of all \emph{non-constant} monomials in the variables $X$ of degree at most $2$. 

We can write $f_i = \sum\limits_{m \in M_X} c_{i,m} m$ for $2\leq i \leq t$ and $f_1 = c + \sum\limits_{m \in M_X} c_{1,m} m$. First, note that we can assume $c \neq 0$; otherwise, the polynomial system $S$ has the trivial solution $(0,\ldots,0)$. Hence, we have $c \in \F^*$ and $c_{i,m} \in \F$.

\begin{lemma}
Define a new variable set $\a$ and new system of polynomial equations $S_1 \cup S_2 \cup S_3$ :
\begin{align*}
    \a &\defeq \{\alpha_0\} \cup \{\alpha_m \mid m \in M_X \} \\
S_1 &\defeq \left\{ c + \sum_{m \in M_X} c_{1,m} \alpha_m = 0, \quad \sum_{m \in M_X} c_{2,m} \alpha_m = 0, \quad \cdots, \quad \sum_{m \in M_X} c_{t,m} \alpha_m = 0 \right\}, \\
S_2 &\defeq \left\{ \alpha_{x_i x_j} - \alpha_{x_i} \alpha_{x_j} = 0 \mid 1 \le i,j \le n \right\}, \\
S_3 &\defeq \left\{ \alpha_0 + \sum_{j=1}^n \alpha_{x_j} = 0 \right\}.
\end{align*}
Then the original system $S$ has a solution if and only if the system $S_1 \cup S_2 \cup S_3$ has a solution.
\end{lemma}

\begin{proof}
Suppose $(u_1, \ldots, u_n)$ is a solution to $S$. Set $\alpha_{x_i} := u_i$, $\alpha_{x_i x_j} := u_i u_j$, and $\alpha_0 := -\sum_j u_j$. Then each equation in $S_1$ is satisfied since the polynomials $f_i$ vanish at $(u_1, \ldots, u_n)$. The constraints in $S_2$ and $S_3$ hold by construction.

Conversely, suppose $(\alpha_0, (\alpha_m)_{m \in M_X})$ satisfies $S_1 \cup S_2 \cup S_3$, and define $u_i := \alpha_{x_i}$. Then $S_2$ implies $\alpha_{x_i x_j} = u_i u_j$, and $S_1$ ensures that $f_i(u_1, \ldots, u_n) = 0$ for all $i$. Thus, $(u_1, \ldots, u_n)$ satisfies $S$.
\end{proof}

The construction of our polynomial for the $\ss$ instance proceeds in two stages. First, we build a polynomial $P_S$ with the following property: $S_1\cup S_2\cup S_3$ has a solution if and only if $P_S$ can be sparsified through shift by a vector satisfying a certain system $\mathcal{L}$ of linear equations. In the second stage, we construct a polynomial $Q_S$ such that $P_S$ can be sparsified through shift by a vector satisfying $\mathcal{L}$ if and only if $Q_S$ has a sparsifying shift. We begin by introducing some new objects and notation required for the construction of the polynomial $P_S$:

Define $N \defeq n^2 + 2n + 1$. Define a new set of variables
  \[
  \b = \left\{ \beta_{k,x_ix_j} \mid 1 \leq k \leq N,\ 1 \leq i,j \leq n \right\}.
  \]
Next, jointly enumerate the following two families of quadratic polynomials:
  \begin{itemize}[noitemsep]
    \item the polynomials from $S_2$, namely $\{ \alpha_{x_ix_j} - \alpha_{x_i}\alpha_{x_j} \mid 1 \leq i,j \leq n \}$, and
    \item the auxiliary polynomials $\{ \beta_{k,x_ix_j} - \alpha_{x_i}\alpha_{x_j} \mid 1 \leq k \leq N,\ 1 \leq i,j \leq n \}$,
  \end{itemize}
  as a single sequence $g_1, \dots, g_r$. Define $g_0$ to be the first polynomial in the system $S_1$, written as
  \[
  g_0 \defeq c + \sum_{m \in M_X} c_{1,m} \alpha_m.
  \]

Introduce three disjoint sets of new variables:
  \begin{itemize}[noitemsep]
    \item $W = \{ w_1, \dots, w_r \}$, one weight variable per polynomial $g_i$,
    \item $Y_1$, a set of $n^2 + n + 1$ variables, and
    \item $Y_2$, a set of $n$ variables.
  \end{itemize}
Define $Y \defeq W \sqcup Y_1 \sqcup Y_2$ and supoose $\gamma_1, \dots, \gamma_r$ are distinct nonzero elements of the field $\F$. This is possible since $r \leq (N + 1)n^2 \leq 4n^4$, and we assume $|\F| > 4n^4$. For any variable set $Z$, we denote by $\lin(Z) \defeq \sum_{z \in Z} z$ the sum of all variables in $Z$.
The polynomial $P_S$ is now defined over the variable set $\a \sqcup \b \sqcup Y$ as below: 
\[
P_S(\a, \b, Y) \defeq 
\underbrace{ \lin(Y_1) \cdot g_0(\a) }_{\text{Part 1}} +
\underbrace{ \sum_{i=1}^r w_i \left[ g_i(\a, \b) + \gamma_i \left( \alpha_0 + \sum_{j=1}^n \alpha_{x_j} \right) \right] }_{\text{Part 2}} +
\underbrace{ \lin(Y_2) \cdot \sum_{i=1}^n \alpha_{x_i}^2 }_{\text{Part 3}}.
\]

\begin{lemma}
\label{lemma:sparsity-shift}
Let $P_S(\a,\b,Y)$ be the polynomial defined above, and consider a shift $(a,b,y) \in \F^{|\a|} \times \F^{|\b|} \times \F^{|Y|}$ of variables. Then, the following are true about the monomial count of $P_S(\a+a, \b+b, Y+y)$:

\begin{enumerate}
    \item \label{item:obs1} If $g_0(a) = 0$, then the monomial count of Part 1 decreases by $|Y_1| = n^2 + n + 1$; otherwise, it remains unchanged.
    \item \label{item:obs2} If $a_{x_1} = \cdots = a_{x_n} = 0$, then the monomial count of Part 3 remains unchanged; otherwise, it increases by some $v \in [2n, (n+1)n]$.
    \item \label{item:obs3} If $a_0 + \sum_{j=1}^n a_{x_j} = 0$, then the monomial count of Part 2 changes by at least $v_1 - v_2$, where
    \[
    v_1 \defeq |\{ i \in [r] : g_i(a,b) \neq 0 \}|, \quad
    v_2 \defeq |\{ j \in [n] : a_{x_j} = \gamma_i \text{ for some } i \in [r] \}|.
    \]
    \item \label{item:obs4} For any $(a,b,y) \in \F^{|\a|} \times \F^{|\b|} \times \F^{|Y|}$, the monomial count of $P_S(\a+a, \b+b, Y+y)$ is \textbf{at least} the monomial count of $P_S(\a+a, \b+b, Y)$.
\end{enumerate}
\end{lemma}

\begin{proof}
We proceed by analyzing how the shift affects the monomial count of each corresponding part of $P_S$.

\smallskip
\noindent\textbf{Monomial count change in Part 1:} Part 1 is 
\[
 \lin(Y_1) \cdot g_0(\a) = \lin(Y_1)\cdot\bigg(c + \sum\limits_{m \in M_X} c_{1,m} \alpha_m\bigg)
\]
After the shift $\a \mapsto \a + a$, it becomes
\[
\lin(Y_1)\cdot\bigg(g_0(a) + \sum\limits_{m \in M_X} c_{1,m} \alpha_m\bigg)
\]

If $g_0(a) = 0$, then the constant term vanishes and the only remaining terms are those with $\a$, so the $|Y_1|$ monomials corresponding to $cy_i$ disappear, causing the monomial count of Part 1 to decrease by exactly $|Y_1|$. Otherwise, these monomials remain, so the monomial count is unchanged. This proves (\ref{item:obs1}).

\smallskip
\noindent\textbf{Monomial count change in Part 3:} Part 3 is
\[
\lin(Y_2) \cdot \sum_{i=1}^n \alpha_{x_i}^2.
\]
Under the shift $\a \mapsto \a + a$, this becomes
\[
\lin(Y_2) \cdot \sum_{i=1}^n (\alpha_{x_i} + a_{x_i})^2 = \lin(Y_2) \cdot \left( \sum_{i=1}^n (\alpha_{x_i}^2 + 2 a_{x_i} \alpha_{x_i} + a_{x_i}^2) \right).
\]

If $a_{x_i} = 0$ for all $i$, then no new terms are introduced, so monomial count remains the same. Otherwise, the terms $2 a_{x_i} \alpha_{x_i}$ and the constants $a_{x_i}^2$ create new monomials, increasing monomial count by at least $2n$ and at most $(n+1)n$, depending on how many $a_{x_i}$ are nonzero. Recall that $|Y_2| = n$. This proves (\ref{item:obs2}).

\smallskip
\noindent\textbf{Monomial count change in Part 2:}  Part 2 is
\[
\sum_{i=1}^r w_i \big[ g_i(\a,\b) + \gamma_i(\alpha_0 + \sum_{j=1}^n \alpha_{x_j}) \big].
\]

Each $g_i$ is of the form $\delta - \alpha_{x_p} \alpha_{x_q}$, where $\delta$ is either $\alpha_{x_p x_q}$ or some $\beta_{k,x_p x_q}$. Under $(\a,\b) \mapsto (\a+a, \b+b)$,  the $i$-th summand becomes
\[
w_i \bigg[ \delta - \alpha_{x_p}\alpha_{x_q} + g_i(a,b) - a_{x_p} \alpha_{x_q} - a_{x_q} \alpha_{x_p} + \gamma_i\Big(\alpha_0 + \sum_{j=1}^n \alpha_{x_j} + a_0 + \sum_{j=1}^n a_{x_j}\Big) \bigg].
\]

Using $a_0 + \sum_j a_{x_j} = 0$, this simplifies to
\[
w_i\bigg[\delta - \alpha_{x_p}\alpha_{x_q} + g_i(a,b) + (\gamma_i - a_{x_p})\alpha_{x_q} + (\gamma_i - a_{x_q})\alpha_{x_p} + \gamma_i\Big(\alpha_0 + \sum\limits_{j \in [n]\backslash \{p,q\}} \alpha_{x_j}\Big)\bigg]
\]
Now, we observe that if $g_i(a,b) \neq 0$, a new monomial $g_i(a,b)w_i$ appears, increasing the monomial count by $1$; these contribute to $v_1$. If $a_{x_p} = \gamma_i$ or $a_{x_q} = \gamma_i$, then the terms $w_i \alpha_{x_q}$ or $w_i \alpha_{x_p}$ vanish, decreasing monomial count; these contribute to $v_2$. Hence, the net monomial count change in Part 2 is at least $v_1 - v_2$. This proves (\ref{item:obs3}).

\smallskip
\noindent\textbf{Proof of (\ref{item:obs4})}: Since $P_S$ is linear in the $Y$ variables, shifting $Y \mapsto Y + y$ can only add new monomials (by expanding linear terms with constants $y$), but cannot remove any existing monomial. Hence,
\[
\text{Monomial count of }(P_S(\a+a, \b+b, Y+y)) \geq \text{Monomial count of }(P_S(\a+a, \b+b, Y)).
\]
This completes the proof.
\end{proof}

\Cref{lem:PsSparsify} below completes the first stage of the proof by showing that the system $S_1 \cup S_2 \cup S_3$ has a solution if and only if the polynomial $P_S$ can be sparsified using shifts that satisfy certain linear constraints.

\begin{lemma} \label{lem:PsSparsify}
    The system $S_1 \cup S_2 \cup S_3$ has a solution if and only if there exists a vector $(a,b,y) \in \F^{|\a|} \times \F^{|\b|} \times \F^{|Y|}$ satisfying the following linear equations:
    \begin{align}
        &\sum_{m \in M_X} c_{2,m} \alpha_m = 0, \cdots, \sum_{m \in M_X} c_{t,m} \alpha_m = 0, \label{eq1} \\
        &\alpha_{x_i x_j} - \beta_{k,x_i x_j} = 0 \quad \text{for all } 1 \leq k \leq N,\; 1 \leq i,j \leq n, \label{eq2} \\
        &\alpha_0 + \sum_{j=1}^n \alpha_{x_j} = 0. \label{eq3}
    \end{align}
such that the shifted polynomial $P_S(\a + a, \b + b, Y + y)$ has strictly fewer monomials than $P_S(\a, \b, Y)$.
\end{lemma}

\begin{proof}
Suppose first that $S_1 \cup S_2 \cup S_3$ has a solution $a \in \F^{|\a|}$. Define $b \in \F^{|\b|}$ by setting $b_{k,x_i x_j} := a_{x_i x_j}$ for all $1 \leq k \leq N$ and $1 \leq i,j \leq n$. Then the shift $(a,b,0)$ reduces the number of monomials in $P_S(\a, \b, Y)$ because:
    \begin{itemize}[noitemsep]
        \item the monomial count of \emph{Part 1} decreases by $n^2 + n + 1$, which follows from (\ref{item:obs1}) in \Cref{lemma:sparsity-shift},
        \item the monomial count of \emph{Part 2} does not increase, because $v_1$ is zero in (\ref{item:obs3}) of \Cref{lemma:sparsity-shift},
        \item the monomial count of \emph{Part 3} increases by at most $(n+1)n$, which follows from (\ref{item:obs2}) in \Cref{lemma:sparsity-shift}.
    \end{itemize}
Also, the vector $(a,b)$ satisfies:
    \begin{itemize}[noitemsep]
        \item \Cref{eq1}, since $a$ satisfies $S_1$,
        \item \Cref{eq2}, by the construction of $b$, and
        \item \Cref{eq3}, since $a$ satisfies $S_3$.
    \end{itemize}

Conversely, suppose there exists a vector $(a,b,y)$ such that $P_S(\a + a, \b + b, Y + y)$ has fewer monomials than $P_S(\a, \b, Y)$ and $(a,b)$ satisfies \Cref{eq1,eq2,eq3}. Using (\ref{item:obs4}) of \Cref{lemma:sparsity-shift}, we may assume without loss of generality that $y = 0$.

    We claim that $a$ satisfies $S_1 \cup S_2 \cup S_3$. Suppose for contradiction that it does not. Since $a$ satisfies the last $(t-1)$ equations of $S_1$ (by \Cref{eq1}) and $S_3$ (by \Cref{eq3}), we must have one of the following cases:

\Case{When $a_{x_p x_q} \neq a_{x_p} a_{x_q}$ for some $p,q \in [n]$}

    Then, by \Cref{eq2}, for all $k \in [N]$, we have $b_{k,x_p x_q} = a_{x_p x_q}$, so $b_{k,x_p x_q} - a_{x_p}a_{x_q} \neq 0$. Hence, (\ref{item:obs3}) of \Cref{lemma:sparsity-shift} implies that the the quantity $v_1$ is at least $N$, and the monomial count of \emph{Part 2} increases by at least $v_1 - v_2 \geq N - n$. By again using (\ref{item:obs1}) and (\ref{item:obs2}) from \Cref{lemma:sparsity-shift}, the monomial count of \emph{Part 1} decreases by at most $n^2 + n + 1$, and that of \emph{Part 3} does not decrease. Thus, the total monomial count increases by at least $N - n - (n^2 + n + 1) \geq 0$, a contradiction!
\smallskip

\Case{When  $g_0(a) \neq 0$ and $a_{x_j} \neq 0$ for some $j \in [n]$} 
    In this case, we have:
    \begin{itemize} [noitemsep]
        \item Monomial count of \emph{Part 1} remains unchanged, which follows from (\ref{item:obs1}) of \Cref{lemma:sparsity-shift}.
        \item Monomial count of \emph{Part 2} decreases by at most $v_2 \leq n$, because $v_1$ is zero in (\ref{item:obs3}) of \Cref{lemma:sparsity-shift}.
        \item Monomial count of \emph{Part 3} increases by at least $2n$, by using (\ref{item:obs2}) of \Cref{lemma:sparsity-shift}.
    \end{itemize}
    Thus, overall monomial count increases by at least $2n - n > 0$, a contradiction!

\smallskip

\Case{When $g_0(a) \neq 0$ and $a_{x_j} = 0$ for all $j \in [n]$}  In this case, we have:
    \begin{itemize}[noitemsep]
        \item Monomial count of \emph{Part 1} and \emph{Part 3} remain unchanged, which follows from (\ref{item:obs1})  and (\ref{item:obs2}) of \Cref{lemma:sparsity-shift}.
        \item Monomial count of \emph{Part 2} increases by $v_1$, because $v_2$ is zero in (\ref{item:obs3}) of \Cref{lemma:sparsity-shift}, as all the $\gamma_i$'s are chosen to be non-zero.
    \end{itemize}

In all the above cases, overall monomial count does not decrease,  a contradiction! Therefore, $a$ must satisfy $S_1 \cup S_2 \cup S_3$.
\end{proof}
We now complete the second stage of the proof. We enumerate the linear constraints from \Cref{eq1,eq2,eq3} as a single list:
\[
L_1(\a,\b) = 0, \quad \dots, \quad L_s(\a,\b) = 0.
\]
By \Cref{lemma:sparsity-shift}, any shift can reduce the monomial count of $P_S$ by at most
\[
(n^2 + n + 1) + v_2 \leq n^2 + 2n + 1.
\]
Let $M := n^2 + 2n + 1$ denote this upper bound. We now introduce $s$ disjoint sets of new variables $Z_1, \dots, Z_s$, each of size $M$, and let $Z \defeq Z_1 \sqcup \cdots \sqcup Z_s$. We now define a new polynomial $Q_S$ for our instance of SparseShift as follows:
\[
Q_S(\a, \b, Y, Z) := P_S(\a, \b, Y) + \sum_{i=1}^s \lin(Z_i) \cdot L_i(\a, \b).
\]
Now we are ready to prove \Cref{sparseshift-thm}, which we restate below for convenience.
\restatablethmsparseshift*
\begin{proof}[Proof of \Cref{sparseshift-thm}]
We know that the system $S$ has a solution if and only if $S_1 \cup S_2 \cup S_3$ has a solution. By \Cref{lem:PsSparsify}, it suffices to prove the following claim:

There exists $(a', b', y', z) \in \F^{|\a|} \times \F^{|\b|} \times \F^{|Y|} \times \F^{|Z|}$ such that $Q_S(\a+a', \b+b', Y+y', Z+z)$ has fewer monomials than $Q_S(\a, \b, Y, Z)$ if and only if there exists $(a, b, y) \in \F^{|\a|} \times \F^{|\b|} \times \F^{|Y|}$ such that:
\begin{align*}
& P_S(\a+a, \b+b, Y+y) \text{ has fewer monomials than } P_S(\a, \b, Y), \\
& L_1(a,b) = \cdots = L_s(a,b) = 0.
\end{align*}
\smallskip
\noindent\textbf{Forward direction}: Suppose there exists $(a', b', y', z)$ such that $Q_S(\a+a', \b+b', Y+y', Z+z)$ has fewer monomials than $Q_S(\a, \b, Y, Z)$. Each $L_i$ is a constant-free linear polynomial. Therefore:
\begin{itemize}[noitemsep]
\item If $L_i(a', b') = 0$, the monomial count of $\lin(Z_i) \cdot L_i(\a, \b)$ remains unchanged after the shift.
\item If $L_i(a', b') \neq 0$, then $L_i(\a+a', \b+b')$ gains a constant term, so its monomial count increases by at least $1$. Consequently, the monomial count of $\lin(Z_i) \cdot L_i(\a, \b)$ increases by at least $|Z_i| = M$.
\end{itemize}
Since the monomial count of $P_S$ can decrease by at most $M$, it follows that to achieve an overall monomial count decrease in $Q_S$, each $L_i$ must vanish under the shift: $L_i(a', b') = 0$ for all $i \in [s]$. Thus, $P_S(\a+a', \b+b', Y+y')$ must have fewer monomials than $P_S(\a, \b, Y)$, and the linear constraints are satisfied.

\smallskip
\noindent\textbf{Reverse direction}: Suppose $(a, b, y)$ satisfies the conditions:
\begin{itemize}[noitemsep]
\item $P_S(\a+a, \b+b, Y+y)$ has fewer monomials than $P_S(\a, \b, Y)$,
\item $L_1(a,b) = \cdots = L_s(a,b) = 0$.
\end{itemize}
Then, for each $i$, $L_i(\a+a, \b+b)$ has the same monomial count as $L_i(\a, \b)$. Therefore, $\lin(Z_i) \cdot L_i(\a, \b)$ is unaffected by the shift. Set $(a', b', y', z) := (a, b, y, 0)$. It follows that $Q_S(\a+a', \b+b', Y+y', Z+z)$ has fewer monomials than $Q_S(\a, \b, Y, Z)$.

Finally, observe that $Q_S$ has at most $\poly(n, t)$ monomials and can be constructed in polynomial time from the input.
\end{proof}
 \section{PolyProj reduces to HN}
\label{sec:res2}
\restatablethmpolyproj*
\begin{proof}[Proof of \Cref{polyproj-thm}] We show that $\proj$ reduces to $\HN$.

We are given a system $S = \{f_1(X) = 0, f_2(X)=0, \cdots, f_t(X)=0\}$ of polynomial equations defined over the variable set $X = \{x_1,\cdots,x_n\}$. Using \Cref{structure-lemma}, we can assume that each $f_i$ is of degree at most $2$.

Introduce new variables $x_0, w_1, \cdots, w_t, z_1, \cdots, z_t$. Our polynomials $f,g$ for the $\proj$ instance will be defined over the variable vector $Y = (x_0, \cdots, x_n, w_1,\cdots,w_t, z_1,\cdots,z_t)$. For $1 \leq i \leq t$, let $d_i \defeq i+1$ and for $1 \leq i \leq n$, let $D_i \defeq i+2(t+1)$. Now, define

\begin{align*}
    &f(Y) = \sum_{i=1}^t w_i^{d_i}(f_i(X) + z_i^{d_i}) + x_0 + \sum_{i=1}^n x_i^{D_i} \hspace{1cm}\text{ and } \\
    &g(Y) = \sum_{i=1}^t w_i^{d_i} z_i^{d_i} + x_0.
\end{align*}

We will show that $S$ has a solution if and only if there exist a matrix $A \in \F^{|Y|\times|Y|}$ and a vector $b \in \F^{|Y|}$ such that $f(AY+b) = g(Y)$. Note that $AY+b$ is simply a vector of affine linear polynomials over $Y$.

If $S$ has a solution $(a_1,\cdots,a_n)$, then the vector\\ $Y' = (x_0-\sum\limits_{i=1}^n a_i^{D_i}, a_1,\cdots,a_n,w_1\cdots,w_t,z_1,\cdots,z_t)$ satisfies
$$f(Y') = \sum\limits_{i=1}^t w_i^{d_i}(0 + z_i^{d_i}) + x_0 -\sum_{i=1}^n a_i^{D_i} + \sum_{i=1}^n a_i^{D_i} = g(Y)\ .$$
Thus, we can find a matrix $A$ and a vector $b$ such that $f(AY+b) = f(Y') = g(Y)$. Note that the resulting matrix $A$ is not invertible as some of the variables are mapped to constants.

Conversely, let $Y' = (L_0, \cdots, L_n, P_1, \cdots, P_t, Q_1, \cdots, Q_t)$ be a vector of affine linear polynomials over $Y$ satisfying $f(Y') = g(Y)$, i.e.,
$$\sum_{i=1}^t P_i^{d_i}(f_i(L_1,\cdots,L_n) + Q_i^{d_i}) + L_0 + \sum_{i=1}^n L_i^{D_i} = \sum_{i=1}^t w_i^{d_i} z_i^{d_i} + x_0\ .$$

Below we claim that such an equality is possible only when each $L_i$ is a constant. The proof is postponed to the end of this section.
\begin{claim} \label{all-const}
    For $1 \leq i \leq n$, $L_i$ is a constant. As a result, $\sum_{i=1}^n L_i^{D_i} = c$ and for all $i \in [t]$, $f_i(L_1,\cdots,L_n) = c_i$ for some constants $c,c_1,\cdots,c_t$.
\end{claim}

By \Cref{all-const}, we have $f(Y') = \sum\limits_{i=1}^t P_i^{d_i}(c_i + Q_i^{d_i}) + L_0 + c$ for the constants $c,c_1,\cdots,c_t$ defined in the claim.

For $1 \leq k \leq t$, define $f^{(k)}(Y') := \sum\limits_{i=1}^k P_i^{d_i}(c_i + Q_i^{d_i}) + L_0 + c$ and $g^{(k)}(Y) := \sum\limits_{i=1}^k w_i^{d_i} z_i^{d_i} + x_0$. Note that $f^{(t)}(Y') = f(Y')$ and $g^{(t)}(Y) = g(Y)$. We make another claim, whose proof is presented at the end of this section.

\begin{claim}\label{pkqk}
    For every $k$, the equality $f^{(k)}(Y') = g^{(k)}(Y)$ implies that $P_kQ_k = w_kz_k$.
\end{claim}

Since we have $f^{(t)}(Y') = f(Y')=g(Y) = g^{(t)}(Y)$, plugging $k=t$ into \Cref{pkqk} gives us $P_tQ_t = w_tz_t$. We show that more is true: in fact, $P_iQ_i = w_iz_i$ for all $i \in [t]$. This is proven using reverse induction on $i$. 

The base case $i=t$ is true, as already mentioned. For the induction step, assume that $P_iQ_i = w_iz_i$ for all $i \geq k+1$. In that case, substituting $w_i=z_i=0$ $\forall i\geq k+1$ into $f(Y')=g(Y)$ yields $f^{(k)}(Y') = g^{(k)}(Y)$. By our previous claim, it follows that $P_kQ_k = w_kz_k$.

Therefore, $P_iQ_i = w_iz_i$ for all $i \in [t]$. In that case, for any $i \in [t]$, $f(Y') = \sum\limits_{i=1}^t P_i^{d_i}(c_i + Q_i^{d_i}) + L_0 + c$ would contain a monomial of degree $d_i$ unless $c_i = 0$. Thus, $f_i(L_1, \cdots, L_n) = c_i = 0$ for all $i \in [t]$ and the constant vector $(L_1,\cdots,L_n)$ is our desired solution of $S$.
\end{proof}

Now we prove the claims made in the above proof. 

\begin{proof}[Proof of \Cref{all-const}]
For sake of contradiction, let $k \in [n]$ be the highest index such that $L_k$ is not a constant. 

Since $L_i$ is a constant for all $i \geq k+1$, $f(Y')-L_k^{D_k}$ has degree at most $D_{k-1} < D_k$. Therefore, if $L_k$ is not a constant, then $f(Y')$ must have degree $D_k$. This is a contradiction as $g(Y)$ has degree $2d_t < D_k$. Hence, our claim is true.
\end{proof}

\begin{proof}[Proof of \Cref{pkqk}]
Suppose $f^{(k)}(Y') = g^{(k)}(Y)$. Notice that $f^{(k)}(Y')-P_k^{d_k}Q_k^{d_k}$ has degree at most $2d_{k-1} < 2d_k$. Therefore, the highest degree term $w_k^{d_k}z_k^{d_k}$ of $g^{(k)}(Y)$ must come from $P_k^{d_k}Q_k^{d_k}$ i.e. $P_kQ_k$ should contain the monomial $w_kz_k$. However, $P_kQ_k$ can not contain any other monomial of degree $2$ because then $f^{(k)}(Y')$ will contain an uncancelled monomial of degree $2d_k$ other than $w_k^{d_k}z_k^{d_k}$, which contradicts with the fact that $w_k^{d_k}z_k^{d_k}$ is the only monomial of degree $2d_k$ in $g^{(k)}(Y)$. Thus, $P_kQ_k = (w_k+a)(z_k+b)$ for some constants $a,b$. We want to show that $a=b=0$. If $a \neq 0$, then $P_k^{d_k}Q_k^{d_k} = (w_k+a)^{d_k}(z_k+b)^{d_k}$ contains the monomial $w_k^{d_k-1}z_k^{d_k}$. This monomial will remain uncancelled in $f^{(k)}(Y')$  as $f^{(k)}(Y')-P_k^{d_k}Q_k^{d_k}$ has degree at most $2d_{k-1} < 2d_k-1$. This leads to a contradiction as $g^{(k)}(Y)$ doesn't contain the monomial $w_k^{d_k-1}z_k^{d_k}$. Thus, $a = 0$. Similarly, $b = 0$. It follows that $P_kQ_k = w_kz_k$, which proves our claim.
\end{proof}

 \section{\texorpdfstring{$\forall \R$-completeness of Deciding Convexity and Biquadratic Non-negativity}{Forall R-hardness of Convexity}}
\label{sec:res4}

In this section, we prove \Cref{bpsd-thm} and \Cref{convexity-thm}. Our proof of \Cref{bpsd-thm} is inspired by \cite[Lemma 2.8]{DBLP:journals/mst/SchaeferS24} where they show $\utr$-hardness of deciding if a degree $6$ polynomial is non-negative, but our construction is significantly more intricate as we want the polynomial produced by the reduction to have the special structure of biquadratic forms. First, we prove a technical lemma in the spirit of \cite[Lemma 2.7]{DBLP:journals/mst/SchaeferS24}. 

\begin{lemma}\label{chain-lemma}
    Let $y_0,\cdots,y_m,z_0,\cdots,z_{m} \in \R$ satisfy the inequality 
    \begin{equation}\label{eq:conineq}
 400\left(\sum_{k=1}^m(y_k-y_{k-1}z_{k-1})^2 + \sum_{k=0}^{m}(y_k-z_k)^2\right) < y_m^2\ .   
    \end{equation}
   
    If $0 < y_0 < 1/2$ and $|z_m| < 1$, then $y_m^2 \leq y_0^{(2^{m+2}+2)/3}$.
\end{lemma}

\begin{proof}
By considering only the second term in the summation on the left-hand side of \ineqref{eq:conineq}, we have
\[
|y_k - z_k| < \frac{|y_m|}{20} \quad \text{for} \quad 0 \leq k \leq m,
\]
and hence,
    \begin{equation}\label{ineq-1}
        y_k - \frac{|y_m|}{20} < z_k < y_k + \frac{|y_m|}{20} \ .
    \end{equation}
    
    In particular, $|y_m|/20 > |y_m-z_m| \geq |y_m|-|z_m|$ implying that $|y_m| < (20/19)|z_m| < 20/19$.

    For $1 \leq k \leq m$, we have
    \begin{align}
        &y_{k-1}z_{k-1}-\frac{|y_m|}{20} < y_k < y_{k-1}z_{k-1}+\frac{|y_m|}{20} \nonumber \\
        \implies &y_{k-1}^2 - y_{k-1}\frac{|y_m|}{20}-\frac{|y_m|}{20} < y_k < y_{k-1}^2 + y_{k-1}\frac{|y_m|}{20}+\frac{|y_m|}{20}\label{ineq-2}
    \end{align}
    where the implication follows from \ineqref{ineq-1}.
    
    Suppose there exists $k \in [m]$ such that $|y_{k-1}| < |y_m|/2$. Then, \ineqref{ineq-2} implies that
    $$ y_k < \frac{|y_m|^2}{4} + \frac{|y_m|^2}{40} + \frac{|y_m|}{20} = \frac{|y_m|}{2}\left(\frac{11}{20}|y_m| + \frac{1}{10}\right) < \frac{|y_m|}{2}$$
    where the last inequality follows from $|y_m| < 20/19$. Similarly,
    $$-y_k < \frac{|y_m|}{2} \left( \frac{1}{10} - \frac{9}{10}|y_m| \right) < \frac{|y_m|}{2}\ .$$
    Thus, $|y_k| < |y_m|/2$ and by induction, we can conclude that $|y_m| < |y_m|/2$, which is a contradiction. Therefore, $|y_{k}| \geq |y_m|/2$ for all $0 \leq k \leq m$.
    
    Now, we will show that $0 < y_i \leq y_0^{(2^{i+1}+1)/3}$ for $0 \leq i \leq m$, which would imply the desired bound on $y_m^2$. The proof is by induction on $i$.

    The base case $i=0$ is trivial.

    Assume the statement is true for $i = k-1$. Then, by the first inequality of \ineqref{ineq-2}, 
    $$ y_{k-1}^2 < y_k + y_{k-1}\frac{|y_m|}{20} + \frac{|y_m|}{20} < y_k + \frac{|y_k|}{5}$$

    where the last inequality follows from $y_{k-1} \leq y_0^{(2^k+1)/3} < 1$ and $|y_m| \leq 2|y_k|$. The above implies that $y_k > 0$. In a similar way, by the second inequality of \ineqref{ineq-2},
    \begin{align*}
        &y_k < y_{k-1}^2 + y_{k-1}\frac{|y_m|}{20}+\frac{|y_m|}{20}  < y_{k-1}^2 + \frac{|y_k|}{5} = y_{k-1}^2 + \frac{y_k}{5}\\
        &\implies y_k < \frac{5}{4}y_{k-1}^2\\
        &\implies y_k < y_0^{-1/3}y_0^{2\cdot(2^k+1)/3} = y_0^{(2^{k+1}+1)/3}
    \end{align*} 
    where the last inequality follows from $y_0^{-1/3} > 2^{1/3} > 5/4$.
\end{proof}

The following result will be used in the proof of \Cref{bpsd-thm} as it actually helps us lower bound the minimum value of a polynomial that only takes positive values on a certain semi-algebraic set.

\begin{theorem}[\cite{SS}, Corollary 3.4]\label{gap}
    If two semi-algebraic sets $S = \{x \in \R^n|\ \phi(x)\}$ and $T = \{x \in \R^n|\ \psi(x)\}$ have positive distance and each of the formulas $\phi$ and $\psi$ have bitlength at most $L\geq 5n$, then the distance between the two sets is at least $2^{-2^{L+5}}$.
\end{theorem}

\restatablethmbpsd*
\begin{proof}[Proof of \Cref{bpsd-thm}]
It suffices to show that the following problem is $\exists\R$-complete: given a biquadratic form $Q \in \R[x_1,\cdots,x_n]$, check whether there exists $(x_1,\cdots,x_n) \in \R^n$ such that $Q(x_1,\cdots,x_n) < 0$. Let's call this problem $\bineg$.

$\bineg$ is trivially in $\etr$. To show its $\etr$-hardness, we reduce the following $\etr$-complete problem to $\bineg$: given a set of quadratic polynomials $\{f_1(x),\cdots,f_t(x)\}$, check if they have a common root $x$ inside the unit ball $B^n(0,1)$ centered at $0$. This problem was shown to be $\etr$-complete in \cite[Lemma 3.9]{Schaefer13}. Further, their construction of these polynomials $f_i$ satisfies an additional property: if these polynomials don't have a common root inside $B^n(0,1)$, then they don't have a common root inside 
$B^n(0,\sqrt{3/2})$ either.

For each $i \in [t]$, let $f_i(x) = \sum_{1 \leq j \leq k \leq n} a_{ijk} x_jx_k + \sum_{1 \leq j \leq n} b_{ij}x_j + c_i$ for constants $a_{ijk},b_{ij},c_i$. Introducing the new variables $x_0$ and  $w_0, w_1, \cdots, w_n$, define for each $i \in [t]$ the bilinear form $g_i(x,w) := \sum_{1 \leq j\leq k \leq n} a_{ijk} x_jw_k + \sum_{1 \leq j \leq n} b_{ij}x_jw_0 + c_ix_0w_0$. Then, there is a natural bijective correspondence between the common roots of $\{f_1(x),\cdots,f_t(x)\}$ and the roots of $g(x,w) := \sum_{i=1}^t g_i^2(x,w) + \sum_{i=1}^n (x_i-w_i)^2 + (x_0-1)^2 + (w_0-1)^2$. A common root $(x_1,\cdots,x_n)$ of the polynomials $f_i$ corresponds to a root $(x_0=1,x_1,\cdots,x_n, w_0=1, w_1=x_1,\cdots,w_n=x_n)$ of $g$.

  Therefore, if the polynomials $f_i$ have a common root inside $B^n(0,1)$, then $g$ has a root $(x,w)\in \R^{2n+2}$ satisfying $\sum_{i=1}^n x_i^2 \leq 1$. 
  
  On the other hand, if they don't have a common root inside $B^n(0,1)$ and hence inside $B^n(0,\sqrt{3/2})$, then $g$ doesn't have a root $(x,w)\in \R^{2n+2}$ satisfying $\sum_{i=1}^n (x_i^2+w_i^2) \leq 3$. Consequently, the compact sets $$S = \{(x,w,g(x,w))\ |\ \sum_{i=1}^n (x_i^2+w_i^2) \leq 3\} \text{ and } T = \{(x,w,0)\ |\ \sum_{i=1}^n (x_i^2+w_i^2) \leq 3\}$$ will have positive distance and by \Cref{gap}, this distance is at least $2^{-2^m}$ for some $m = O(L+n)$, where $L$ is the bit complexity of $g$. Thus, when the polynomials $f_i$ don't have a common root inside $B^n(0,1)$, we have
  \begin{equation}
      \min\limits_{\substack{(x,w) \in \R^{2n+2}:\\
  \sum_{i=1}^n (x_i^2+w_i^2) \leq 3}} g(x,w) \geq 2^{-2^m}\ . \label{min-value-g}
  \end{equation}

Now, consider the inequality
\begin{align}
    &g(x,w) + \left(\sum_{i=1}^{n+1} x_iw_i - 1\right)^2 + (x_{n+1}-w_{n+1})^2 \nonumber\\
     &+ 400\left((y_1-1/16)^2+\sum_{k=2}^m(y_k-y_{k-1}z_{k-1})^2 + \sum_{k=1}^{m}(y_k-z_k)^2\right)
    &< 2y_mz_m - z_m^2 - y_m^2z_m^2 \ . \label{main-ineq}
\end{align}

If the polynomials $f_i$ have a common root in $B^n(0,1)$ and hence $g$ has a root $(x_0,\cdots,x_n,w_0,\cdots,w_n)$ 
satisfying $\sum_{i=1}^n x_i^2 \leq 1$, then this $(x_0,\cdots,x_n,w_0,\cdots,w_n)$ along with $x_{n+1}=w_{n+1} = \sqrt{1-\sum_{i=1}^n x_i^2}$ and $y_k = z_k = 4^{-2^k}$ for $1\leq k \leq m$ will satisfy \ineqref{main-ineq}, as the LHS will be zero and the RHS will be positive.

Conversely, assume that \ineqref{main-ineq} holds for some $(x,w,y,z) \in \R^{n+2}\times\R^{n+2}\times\R^{m+1}\times\R^{m+1}$. Since the LHS is a sum of squares, we have
\begin{align*}
    &2y_mz_m - z_m^2 - y_m^2z_m^2 > 0\\
    \implies &z_m\left(z_m - \frac{2y_m}{1+y_m^2}\right) < 0\\
    \implies &|z_m| < 1 \hspace{3cm} \text{as $2y_m/(1+y_m^2) \in [-1,1]$ for all $y_m \in \R$.}
\end{align*}

Since $\rhs = y_m^2 - (y_m-z_m)^2 - y_m^2z_m^2 \leq y_m^2$, we can now use \Cref{chain-lemma} with $y_0 = z_0 = 1/4$ to conclude that $y_m^2 \leq (1/4)^{(2^{m+2}+2)/3} = 2^{-4(2^{m+1}+1)/3} < 2^{-2^m}$.
 
Therefore, we have $g(x,w) + \left(\sum_{i=1}^{n+1} x_iw_i - 1\right)^2 + (x_{n+1}-w_{n+1})^2 < \rhs \leq y_m^2 < 2^{-2^m}$. We can bound each of the three summands here by $2^{-2^m}$, which yields the following:
\begin{align*}
    \sum_{i=1}^{n+1} (x_i^2 + w_i^2) &=\sum_{i=1}^n(x_i-w_i)^2 + (x_{n+1}-w_{n+1})^2 + 2\sum_{i=1}^{n+1} x_iw_i \\
    &\leq g(x,w) + (x_{n+1}-w_{n+1})^2 + 2\left(\sum_{i=1}^{n+1} x_iw_i-1\right) + 2\\
    &< 2^{-2^m} + 2^{-2^m} + 2\cdot 2^{-2^{m-1}} + 2 < 3 \ .
\end{align*} 

Thus, there exists $(x,w)$ such that $g(x,w) < 2^{-2^m}$ and $\sum_{i=1}^{n} (x_i^2 + w_i^2) < 3$. Using \ineqref{min-value-g}, we conclude that the polynomials $f_i$ have a common root in $B^n(0,1)$.

We have shown that the polynomials $f_i$ have a common root in $B^n(0,1)$ if and only if $\exists (x,w,y,z) \in \R^{n+2}\times\R^{n+2}\times\R^{m+1}\times\R^{m+1}$ such that
\begin{align*}
    & h(x,w,y,z) := g(x,w) + \left(\sum_{i=1}^{n+1} x_iw_i - 1\right)^2 + (x_{n+1}-w_{n+1})^2 \nonumber\\
     &+ 400\left((y_1-1/16)^2+\sum_{k=2}^m(y_k-y_{k-1}z_{k-1})^2 + \sum_{k=1}^{m}(y_k-z_k)^2\right)
    - (2y_mz_m - z_m^2 - y_m^2z_m^2) < 0 \ .
\end{align*}

For the final step, it remains to turn $h$ into a biquadratic form $Q$. 
Define $$Q(x,w,y,z,\alpha, \beta) = \alpha^2 \beta^2 \cdot h(x/\alpha,w/\beta,y/\alpha,z/\beta)$$ 
be the homogenization of the polynomial $h(x,w,y,z)$ using two new variables
$\alpha$ and $\beta$. Note that $Q$ is indeed a biquadratic form with respect to 
the partition of the  variable set into the two groups $(x,y,\alpha)$ and $(w,z,\beta)$. 
    
    If $\exists (x,w,y,z)$ such that $h(x,w,y,z) < 0$, then we have $Q(x,w,y,z,1,1) < 0$. 

    Conversely, if $\exists (x,w,y,z,\alpha,\beta)$ such that $Q(x,w,y,z,\alpha,\beta) < 0$, 
    then $\alpha$ and $\beta$ are nonzero as\\ $Q(x,w,y,z,\alpha,\beta)$ is a sum of squares
     when $\alpha=0$ or $\beta=0$. Therefore, we have $h(x/\alpha,w/\beta,y/\alpha,z/\beta) = \frac{1}{\alpha^2\beta^2}Q(x,w,y,z,\alpha,\beta) < 0$.
\end{proof}

\begin{remark}
The result of \Cref{bpsd-thm} easily generalizes to show that deciding non-negativity of triquadratic forms is also $\forall\mathbb{R}$-complete. Indeed, one can multiply the biquadratic form $Q(x,w,y,z,\alpha,\beta)$ constructed above by the square of a variable from a third group to obtain a triquadratic form $T$. More generally, the same construction extends to show that testing non-negativity of $k$-quadratic forms is $\forall\mathbb{R}$-complete for any $k \geq 2$.
\end{remark}

Now, we prove \Cref{convexity-thm}. We use the following reduction from $\bineg$ to deciding convexity of quartic polynomials shown in \cite{AOPT}.

\begin{theorem}[\cite{AOPT}, Theorem 2.3] Let $b(X,Y)$ be a biquadratic form with respect to the partition $X \sqcup Y$ of the variable set and let $|X|=|Y|=n$. Define the $n\times n$ polynomial matrix $C(X,Y)$ by setting 
\[
  [C(X,Y)]_{ij} = \frac{\partial b(X,Y)}{\partial X_iY_j}
\]
and let $\gamma$ be the largest coefficient, in absolute value, of any monomial present in some entry of the matrix $C(X,Y)$. Let $f$ be the form given by 
$$f(X,Y) := b(X,Y) + \frac{n^2\gamma}{2}\left(\sum_{i=1}^n (X_i^4+Y_i^4) + \sum_{1\leq i < j \leq n} (X_i^2X_j^2 + Y_i^2Y_j^2) \right)\ .$$
Then, $b(X,Y)$ is non-negative if and only if $f(X,Y)$ is convex.
\end{theorem}

Using this theorem along with \Cref{bpsd-thm}, we easily get that deciding if a quartic polynomial is convex is $\utr$-hard. It only remains to show that this problem is also in $\utr$.

A polynomial $f \in \R[x_1,\cdots,x_n]$ is convex\\
$\iff \forall x \in \R^n \ H_f(x) \succeq 0$\\
$\iff \forall x,z \in \R^n \ z^T H_f(x) z \geq 0.$

Since we can compute the polynomial entries of $H_f(x)$ in polynomial time, we can decide in $\utr$ if $f$ is convex.

 \section{\texorpdfstring{$\forall \R$-completeness of Hyperbolicity and Real Stability}{Forall R-completeness of Hyperbolicity and Real Stability}}
\label{sec:res3}

First, we prove \Cref{hyperbolic-thm} which states that deciding if a homogeneous quartic polynomial $p$ is hyperbolic with respect to a vector $e$ is $\utr$-complete.
    
\textbf{Deciding hyperbolicity is in $\utr$}: \hspace{0.5cm}
It is enough to show that the following problem is in $\etr$: given a quartic homogeneous polynomial $p(x_1,\cdots,x_n)$ and a vector $e\in \R^n$, decide if $p$ is \textbf{not} hyperbolic with respect to $e$.

From the definition of hyperbolicity, we have the following.\\
$p$ is not hyperbolic with respect to $e$\\
$\iff \Big(p(e) \leq 0\Big) \vee \Big(\exists x \in \R^n \ \exists t \in \C\setminus \R \ \ \ p(x+te) = 0\Big)$\\
$\iff \Big(p(e) \leq 0\Big) \vee \Big(\exists x \in \R^n \ \exists (y,z) \in \R^2 \ \ \ z \neq 0\ \wedge\ p(x+ye+ize) = 0\Big)$\\
$\iff \exists x \in \R^n \ \exists (y,z) \in \R^2 \ \ \Big(p(e) \leq 0\Big) \vee \Big(z \neq 0\ \wedge\ \Re\big(p(x+ye+ize)\big) = 0\ \wedge\ \Im\big(p(x+ye+ize)\big) = 0\Big)$\\
where $\Re(f)$ and $\Im(f)$ denote the real and imaginary parts of the polynomial $f$ respectively. The real and imaginary parts of $p(x+ye+ize)$ can be computable in polynomial time since $p$ is of degree only $4$.

Thus, we can decide in $\etr$ if $p$ is not hyperbolic with respect to $e$.

\textbf{Deciding hyperbolicity is $\utr$-hard}: \hspace{0.5cm}
By \Cref{bpsd-thm}, it suffices to reduce the problem of deciding non-negativity of a biquadratic polynomial to the problem of deciding hyperbolicity of a homogeneous quartic polynomial with respect to a vector. 

Given a biquadratic form $Q(x_1, \cdots, x_n)$, we will construct a homogeneous polynomial $p$ of degree $4$ and a vector $e$ such that $p$ is hyperbolic with respect to $e$ iff
$\forall x_1, \cdots, x_n \in \R$, $Q(x_1,\cdots,x_n) \geq 0$.

Let $\norm{x} = \sqrt{x_1^2+\cdots+x_n^2}$. Choose $\beta$ large enough so that 
 $$\forall x_1, \cdots, x_n \in \R,\  Q(x_1,\cdots,x_n) \leq \frac{\beta^2}{4}\norm{x}^4\ .$$ It's easy to see that $\beta = 2n^2C$ satisfies this constraint where $C$ is the largest absolute value of a coefficient of $Q$.

Now, we define our polynomial $$p(x_0,\cdots,x_n) = x_0^4 - \beta x_0^2 \norm{x}^2 + Q(x_1,\cdots,x_n)$$ and vector $e = (1,0,\cdots,0) \in \R^{n+1}$. Note that $p(e)>0$, which is one of the conditions for hyperbolicity.

Let $W$ be the hyperplane in $\R^{n+1}$ given by the equation $x_0=0$. Note that $e \not\in W$. Hence, by \Cref{subspace}, $p$ is hyperbolic with respect to $e$ if and only if $B_{p,e}(x) \succeq 0$ for all $x \in W$. We compute this $B_{p,e}(x)$ for $x \in W$:

When $x \in W$, $p(x+te) = p(t,x_1,\cdots,x_n) = t^4 - \beta t^2 \norm{x}^2 + Q(x_1,\cdots,x_n)$

and $D_ep(x+te) = \frac{d}{du} p(x+te+ue) \vert_{u,x_0=0} = 4t^3  - 2\beta t \norm{x}^2$. As a result, we have
\begin{align*}
     &\frac{p(x+te)D_ep(x+se) - p(x+se)D_ep(x+te)}{t-s} \\
    &= 2\beta\norm{x}^2Q + (2\beta^2\norm{x}^4-4Q)st -4Q(s^2+t^2)+2\beta\norm{x}^2(s^2t^2-s^3t-st^3) + 4s^3t^3\ .
\end{align*}
Hence, $$B_{p,e}(x) = \begin{bmatrix}
2\beta\norm{x}^2Q &  & -4Q & \\
 & 2\beta^2\norm{x}^4-4Q & & -2\beta\norm{x}^2\\
-4Q &  & 2\beta\norm{x}^2 & \\
 & -2\beta\norm{x}^2 &  & 4
\end{bmatrix} \ .$$
Computing the Schur complement (with respect to the bottom-right $2\times 2$ block) of $B_{p,e}(x)$ and dividing by $4$, we obtain
\begin{align*}
    B_{p,e}(x) \succeq 0 &\iff \begin{bmatrix}
\frac{2Q}{\beta\norm{x}^2}(\frac{\beta^2}{4}\norm{x}^4-Q) &   \\
 &  & \frac{\beta^2}{4}\norm{x}^4-Q
\end{bmatrix} \succeq 0\\
&\iff Q \geq 0 \text{ and } \frac{\beta^2}{4}\norm{x}^4-Q \geq 0 \ .
\end{align*}
Here we have used the fact that $B_{p,e}(x) \succeq 0$ if and only if its Schur complement is $\succeq 0$. Thus, $p$ is hyperbolic with respect to $e$\\
$\iff \forall x_1, \cdots, x_n \in \R$, $0 \leq Q(x_1,\cdots,x_n) \leq \frac{\beta^2}{4}\norm{x}^4$\\
$\iff \forall x_1, \cdots, x_n \in \R$, $Q(x_1,\cdots,x_n) \geq 0$\\
as the upper bound is always satisfied by our choice of $\beta$.

Now, we prove \Cref{stability-thm} which states that deciding if a homogeneous quartic polynomial $p$ is real stable is $\utr$-complete.

\textbf{Deciding real stability is in $\utr$}: 

$p$ is real stable\\
$\iff \forall z \in \C^n \ \left(\bigwedge_{k=1}^n \Im(z_k) > 0\right)\implies p(z) \neq 0$\\
$\iff \forall x,y \in \R^n \ \left(\bigwedge_{k=1}^n y_k > 0\right)\implies \left(\Re(p(x+iy)) \neq 0 \vee \Im(p(x+iy)) \neq 0\right)$\\
where the real and imaginary parts of $p(x+iy)$ are computable in polynomial time since $p$ is of degree only $4$.

Thus, we can decide in $\utr$ if $p$ is real stable.

\textbf{Deciding real stability is $\utr$-hard}: \hspace{0.5cm} Saunderson \cite{saunderson} showed $\coNP$-hardness of deciding if a cubic homogeneous polynomial $p$ is hyperbolic with respect to a vector $e$. Chin \cite{chin2024} demonstrated a way to leverage this result to prove $\coNP$-hardness of deciding real stability of a polynomial. For the polynomial $p$ and the vector $e$ in the construction of \cite{saunderson}, they showed how to build a cubic homogeneous polynomial $\tilde{p}$ such that $p$ is hyperbolic with respect to $e$ if and only if $\tilde{p}$ is real stable.

Using the exact same method as in \cite{chin2024}, we can leverage the $\utr$-hardness result in \Cref{hyperbolic-thm}  to obtain $\utr$-hardness of deciding if a quartic homogeneous polynomial is real stable. Basically, we can show the following:

\begin{theorem}\label{hyperbolic-to-stable}
    Let $Q$ be a quartic homogeneous polynomial and $C$ be the largest absolute value of a coefficient in $Q$. Further, define the vector $e_0 = (1,0,\cdots,0) \in \R^{n+1}$ and the polynomial $p(x_0,x) = x_0^4 - \beta x_0^2 \norm{x}^2 + Q(x_1,\cdots,x_n)$ for $\beta = 2n^2C$. Then, we can construct a quartic homogeneous polynomial $\tilde{p}$ such that $\tilde{p}$ is real stable if and only if $p$ is hyperbolic with respect to $e_0$. Moreover, $\tilde{p}$ can be constructed in polynomial time in the total bit complexity of $p$.
\end{theorem}

We can prove this theorem by replicating \cite[Section 3]{chin2024} with minute changes and the proof is given in \Cref{app:stability}. The polynomial $p$ and the vector $e_0$ in this theorem describe our construction in the $\utr$-hardness result of \Cref{hyperbolic-thm}. Thus, we can conclude that it is $\utr$-hard to decide if a quartic homogeneous polynomial is stable.

\begingroup
\bibliography{references}
\bibliographystyle{plain}

\appendix

\section{Deciding Real Stability Reduces to Deciding Hyperbolicity}
\label{app:stability}

We show how to prove \Cref{hyperbolic-to-stable} by using the strategy of \cite[Section 3]{chin2024}. As described in \Cref{hyperbolic-to-stable}, let $Q$ be a quartic homogeneous polynomial and $C$ be the largest absolute value of a coefficient in $Q$. Define the polynomial $p(x_0,x) = x_0^4 - \beta x_0^2 \norm{x}^2 + Q(x_1,\cdots,x_n)$ for $\beta = 2n^2C$. For $i \in [n+1]$, let $e_i$ denote the vector in $\R^{n+1}$ with one in the $i$-th coordinate and zeros in the rest.

\begin{lemma}\label{eps-positivity}
    Let the polynomials $p,Q$ and the constants $C,\beta$ be as defined above. If $0 < \eps < \min\{\frac{1}{\sqrt{2\beta}},\frac{1}{2n}\}$, then $p(\norm{x}, \eps x) > 0$ for all $x \in \R^n\setminus\{0\}$.
\end{lemma}
\begin{proof}
    Fix $x \in \R^n \setminus \{0\}$. We have 
    $$p(\norm{x}, \eps x) = \norm{x}^4 - \beta \eps^2 \norm{x}^4 + Q(\eps x) = \norm{x}^4(1-\beta \eps^2 + \eps^4 Q(x/\norm{x}))\ .$$
    Since $Q$ is a quartic homogeneous polynomial, it has at most $n^4$ monomials and every monomial of $Q$ has absolute value $\leq 1$ at the point $x/\norm{x}$. Therefore, $Q(x/\norm{x}) \geq -Cn^4$.

    Thus, when $0 < \eps < \min\{1/\sqrt{2\beta},1/(2n)\}$,
    \begin{align*}
        p(\norm{x}, \eps x) &\geq  \norm{x}^4(1-\beta \eps^2 - \eps^4 Cn^4)\\
        &\geq \norm{x}^4\left(1-\beta\cdot \frac{1}{2\beta} - Cn^4 \cdot\min\left\{\frac{1}{4. 4n^4C^2},\frac{1}{16n^4}\right\}\right)> 0 \ .
    \end{align*}
\end{proof}

\begin{lemma}\label{K-lemma}
    Let $p$ be as defined above. For $0 < \eps < \min\{\frac{1}{\sqrt{2\beta}},\frac{1}{2n}\}$, define $K_{\eps} = \text{cone}\{e_0\pm\eps e_i\ |\ i \in [n]\}$. Then, $p$ is hyperbolic with respect to $e_0$ if and only if it is hyperbolic with respect to $K_{\eps}$.
\end{lemma}
\begin{proof}
    One direction is trivial as $e_0 \in K_{\eps}$. 
    
    For the other direction, assume that $p$ is hyperbolic with respect to $e_0$. Then, $p$ is hyperbolic with respect to all the vectors in the hyperbolicity cone $\Lambda_+(p,e_0)$. We want to show that the cone $K_{\eps}$ is contained in the cone $\Lambda_+(p,e_0)$. Since $p(e_0) > 0$ and by \Cref{garding}, $\Lambda_+(p,e_0)$ is exactly the connected component of $\{(x_0,x) \in \R^{n+1}: p(x_0,x) \neq 0\}$ containing $e_0$, it suffices to show that $p(x_0,x) > 0$ on $K_{\eps}$.

    For all $(x_0,x) \in K_{\eps}$, we have $\norm{x} = \delta x_0$ for some $\delta \leq \eps$. Hence, by \Cref{eps-positivity}, $p(x_0,x) > 0$ for all $(x_0,x)\in K_{\eps}$.
\end{proof}

\begin{theorem}
    Let the polynomial $p$ and the constant $\beta$ be as defined above. For $0 < \eps < \min\{\frac{1}{\sqrt{2\beta}},\frac{1}{2n}\}$, let $M$ be the $(n+1)\times 2n$ matrix whose columns are $e_0\pm\eps e_i$ for $i \in [n]$. Define 
    $$\tilde{p}(x_1,\cdots,x_{2n}) = p(Mx) \ .$$
    Then, $p$ is hyperbolic with respect to $e_0$ if and only if $\tilde{p}$ is real stable.
\end{theorem}

\begin{proof}
    Let $\tilde{p}$ be real stable. We want to show that for every $z \in \R^{n+1}$, $p(te_0+z)$ has only real roots. Since $M$ has full row rank, there exists $x \in \R^{2n}$ such that $z = Mx$. Then,
    $$p(te_0+z) = p\paren{\frac{t}{2n} M\mathbbm{1}+Mx} = \tilde{p}\paren{\frac{t}{2n} \mathbbm{1} + x}$$
    where $\mathbbm{1}$ is the vector of all ones.  Since $\frac{1}{2n} \mathbbm{1} \in \R^{2n}_{>0}$ and $\tilde{p}$ is real stable, $p(te_0+z)$ has only real roots and hence $p$ is hyperbolic with respect to $e_0$. 

    Conversely, let $p$ be hyperbolic with respect to $e_0$. By \Cref{K-lemma}, $p$ is hyperbolic with respect to $K_{\eps}$. Then, for any $v \in \R^{2n}_{>0}$ and $x \in \R^{2n}$,
    $$\tilde{p}(tv+x) = p(t(Mv)+Mx)$$
    which has only real roots since $Mv \in K_{\eps}$. Therefore, $\tilde{p}$ is real stable.
\end{proof} \section{Machine Models for the Existential Theory of the Reals}
\label{app:real-ram}

\cite{DBLP:journals/siamcomp/EricksonHM24} expands the concept of word RAMs to encompass real computations. The real RAMs introduced in \cite{DBLP:journals/siamcomp/EricksonHM24} facilitate both integer and real computations simultaneously, unlike the BSS model of real computation \cite{blum1989theory}. The real RAMs in \cite{DBLP:journals/siamcomp/EricksonHM24} are more convenient to program, as they include features such as indirect memory access to real registers, which are crucial for implementing algorithms over real numbers in an efficient manner.

We give a brief outline how the real RAMs work:
A real RAM gets two input vectors: one vector of integers and one of real numbers. It has two types of registers: word registers and real registers. Word registers store integers with $w$ bits, where $w$ is the word size. The total number of registers of each type is $2^w$, since we can address at most that many registers with $w$ bits. Arithmetic operations and bitwise Boolean operations are performed on word registers, which interpret words as integers between $0$ and $2^w-1$. Real registers, however, only support arithmetic operations. Word registers can indirectly address both types of registers, and control flow is managed through conditional jumps based on comparisons between word registers or between a real register and the constant 0. For more details, refer to the original paper \cite{DBLP:journals/siamcomp/EricksonHM24}. 

Real RAMs characterize the existential theory of the reals. The main result of \cite{DBLP:journals/siamcomp/EricksonHM24} is that a problem belongs to $\exists\R$ if and only if there exists a polynomial-time real verification algorithm for it, given as a real RAM. Thus, real RAMs allows us to prove that a problem is in $\exists\R$ by giving an algorithm. Besides the input $x$, which is a sequence of words, the real verification algorithm receives a certificate $c$ comprising a sequence of real numbers and another sequence of words. The input $x$ is contained in the language recognized by the verifying algorithm if there exists a certificate $c$ that makes the real verification algorithm accept. Conversely, $x$ is not in the language if the real verification algorithm rejects all certificates. One of the main results in \cite{DBLP:journals/siamcomp/EricksonHM24} is the following:

\begin{theorem}
$L \in \exists\R$ if and only if there is a polynomial time real verification algorithm for $L$.    
\end{theorem} 
\endgroup

\end{document}